\documentclass{vldb}

\usepackage{algorithm}
\usepackage{times}
\usepackage{color}
\usepackage{url}
\usepackage{subfigure}
\usepackage{xspace}
\usepackage[noend]{algorithmic}
\usepackage{enumerate}
\usepackage{multirow}
\usepackage{balance} 
\newcommand{\squishlist}{
   \begin{list}{$\bullet$}
    {
      \setlength{\itemsep}{0pt}
      \setlength{\parsep}{3pt}
      \setlength{\topsep}{3pt}
      \setlength{\partopsep}{0pt}
      \setlength{\leftmargin}{1.5em}
      \setlength{\labelwidth}{1em}
      \setlength{\labelsep}{0.5em} } }

\newcommand{\squishend}{
    \end{list}  }

\newcommand{\eat}[1]{}

\newcommand{\aset}{\mathcal{S}}
\newcommand{\discbudget}{{\sc DiscBudget}\xspace}
\newcommand{\tradeoff}{{\sc TradeOff}\xspace}
\newcommand{\sparsi}{{SPARSI}\xspace}

\newtheorem{definition}{Definition}
\newtheorem{example}{Example}
\newtheorem{theorem}{Theorem}

\begin{document}

\title{On Sharing Private Data with Multiple Non-Colluding Adversaries}
\numberofauthors{3}
\author{
            \alignauthor Theodoros Rekatsinas\\
            \affaddr{University of Maryland} 
                \email{thodrek@cs.umd.edu}
            \alignauthor Amol Deshpande\\
            \affaddr{University of Maryland} 
                \email{amol@cs.umd.edu}
            \alignauthor Ashwin Machanavajjhala \\
            \affaddr{Duke University} 
                \email{ashwin@cs.duke.edu}
}
\maketitle

\begin{abstract}
We present \sparsi, a novel theoretical framework for partitioning sensitive data across multiple non-colluding adversaries. Most work in privacy-aware data sharing has considered disclosing summaries where the aggregate information about the data is preserved, but sensitive user information is protected. Nonetheless, there are applications, including online advertising, cloud computing and crowdsourcing markets, where detailed and fine-grained user-data must be disclosed. We consider a new data sharing paradigm and introduce the problem of {\em privacy-aware data partitioning}, where a sensitive dataset must be partitioned among {\em k} untrusted parties ({\em adversaries}). The goal is to maximize the utility derived by partitioning and distributing the dataset, while minimizing the total amount of sensitive information disclosed. The data should be distributed so that an adversary, without colluding with other adversaries, cannot draw additional inferences about the private information, by {\em linking} together multiple pieces of information released to her.  The assumption of no collusion is both reasonable and necessary in the above application domains that require release of private user information. \sparsi enables us to formally define privacy-aware data partitioning using the notion of  {\em sensitive properties} for modeling private information and a {\em hypergraph} representation for describing the interdependencies between data entries and private information. We show that solving privacy-aware partitioning is, in general, NP-hard, but for specific information disclosure functions, good approximate solutions can be found using {\em relaxation} techniques. Finally, we present a local search algorithm applicable to generic information disclosure functions. We apply \sparsi together with the proposed algorithms on data from a real advertising scenario and show that we can partition data with no disclosure to any single advertiser.

\end{abstract}




\section{Introduction}
\label{sec:intro}
The landscape of online services has changed significantly in the recent years. More and more sensitive information is released on the Web and is processed by online services.  The most common paradigm to consider are people who rely on online social networks to communicate and share information with each other. This leads to a diverse collection of voluntarily published user data. Online services such as Web search, news portals, recommendation and e-commerce systems,  collect and store this data in their effort to provide high-quality personalized experiences to a heterogeneous user base. Naturally, this leads to increased concerns related to an individual's privacy and the possibility of private personal information being aggregated by untrusted third-parties such as advertisers.

A different application domain that is increasingly popular is crowdsourcing markets. Tasks, typically decomposed into micro-tasks, are submitted by users to a crowdsourcing market and are fulfilled by a collection of workers. The user needs to provide each worker with the necessary data to accomplish each micro-task. However, this data may contain information that is sensitive and care must be taken not to disclose any more sensitive information than minimally needed to accomplish the task. Consider, for example, the task of labeling a dataset that contains information about the location of different individuals, that needs to be used as input to a machine learning algorithm.  Since the cost of hand-labeling the dataset is high, submitting this task to a crowdsourcing market provides an inexpensive alternative. However,  the dataset might contain sensitive information about the trajectories the individuals follow as well as the structure of the social network they form. Hence, we must perform a clever partitioning of the dataset to the different untrusted workers in order to avoid disclosing sensitive information. Observe, that under this paradigm, the sensitive information contained in the dataset is not necessarily associated with a particular data entry.

Similarly with the rise of cloud computing, increasing volumes of private data are being stored and processed on untrusted servers in the cloud. Even if the data is stored and processed in an encrypted form, an adversary may be able to infer some of the private information by aggregating, over a period of time, the information that is available to it  (e.g., password hashes of users, workload information). This has led security researchers to recommend splitting data and workloads across systems or organizations to remove such points of compromise.

In all applications presented above, a party, called {\em publisher}, is required to distribute a collection of data (e.g., user information) to many different third parties. The utility in sharing data results either from the improved quality of personalized services or from the cost reduction in fulfilling a decomposable task. The sensitive information is often not limited to the identity of a particular entity in the dataset (e.g., a user using a social network based service), but rather arises from the combination of a set of data items. It is these sets we would like to partition accross different adversaries. We next use two real-world examples to illustrate this.
\newpage
\begin{example}
\label{ex:socialnetwork}
Consider a location based social network, such as Gowalla\footnote{\url{http://en.wikipedia.org/wiki/Gowalla}} and Brightkite\footnote{\url{http://en.wikipedia.org/wiki/Brightkite}}, where users check-in at different places they visit. The available data contains information about the locations of the users at different time instances and the structure of the social network connecting the users. User location data is of particular interest to advertisers, as analyzing it can provide them with a rather detailed profile of the habits of the user. Using such data allows advertisers to devise highly efficient personalized marketing strategies. Hence, they are willing to pay large amounts of money to the data publisher for user information. However, analyzing the location of multiple users collectively can reveal information about the friendship links between users, thus, revealing the structure of the social network \cite{cho:2011}. Disclosing the structure of the social network might not be desirable by the online social network provider, as it can be used for viral marketing purposes, which may drive users away from using the social network. It is easy to see that a natural tradeoff exists between publishing user data, and receiving high monetary utility, versus keeping this data private to ensure the popularity of the social network.
\end{example}

This example shows how an adversary may infer some sensitive information that is not explicitly mentioned in the dataset but is related to the provided data and can be inferred only when particular entries of the dataset are collectively analyzed. Not revealing all those entries together to an adversary prevents disclosure of the sensitive information. We further exemplify this setup using a crowdsourcing application.
\begin{example}
\label{ex:crowdsourcing}
Consider a data publisher with a collection of medical prescriptions to be transcribed.  Each prescription contains sensitive information, such as the disease of the patient, the prescribed medication, the identity of the patient, and the identity of the doctor. Furthermore, the publisher would like to minimize the total cost of the transcription. Thus, she considers using a crowdsourcing solution where she partitions the task into micro-tasks to be submitted to multiple workers. It is obvious that if all fields in the prescription are revealed to the same worker, highly sensitive information is disclosed. However, if the dataset is partitioned in such a way that different workers are responsible for transcribing different fields of one prescription, no information is disclosed as patients cannot be linked with a particular disease or a particular doctor. In this case, the utility of the publisher stems from fulfilling the task at a reduced cost.
\end{example}

Despite being simplistic, the second example illustrates how distributing a dataset can allow one to use it for a particular task, while minimizing the disclosure of sensitive information. Motivated by applications such as the ones presented above, we introduce the problem of \emph{privacy-aware partitioning} of a dataset, where our goal is to partition a dataset among $k$ untrusted parties and to maximize either user's utility, or the third parties' utilities, or a combination of those. Further, we would like to do this while minimizing the total amount of sensitive information disclosed.

Most of the previous work has either considered sharing \emph{pri-vacy-preserving summaries} of data, where the aggregate information about the population of users is preserved, or has bypassed the use of personal data and its disclosure to multiple advertisers~\cite{sweeney02:kAnon, ashwin06:ldiversity, hardt10:simple}. These approaches focus on worst-case scenarios assuming arbitrary collusion among adversaries. Therefore, all adversaries are combined and treated as a single adversary.  However, this strong threat model does not allow publishing of fine-grained information. Other approaches have explicitly focused on online advertising, and have developed specialized systems that limit the disclosure of sensitive user-related information when deployed to a user's Web browser~\cite{guha:2011,toubiana:2010}. Finally, Krause et al. have studied how the disclosure of a subset of the attributes of a data entry can allow access to fine-grained information~\cite{krause:2008}. While they examine the utility and disclosure tradeoff, their proposed framework does not take into account the interdependencies across different data entries and assumes a single adversary (third party).\footnote{For the remainder of the paper we treat adversaries and third parties as the same}

In this work we propose \sparsi a new framework that allows us to formally reason about leakage of sensitive information in scenarios such as the ones presented above, namely, setups where we are given a dataset to be partitioned among a set of non-colluding adversaries in order to obtain some utility. We consider a generalized form of utility that captures both the utility that each adversary obtains by receiving part of the data and the user's personal utility derived by fulfilling a task. We elaborate more on this generalization in the next section. This raises a natural tradeoff between maximizing the overall utility while minimizing information disclosure. We provide a formal definition of the privacy-aware data partitioning problem, as an optimization of the aforementioned tradeoff. 

While non-collusion results in a somewhat weaker threat model, we argue that it is a reasonable and practical assumption in a variety of scenarios, including the ones discussed above. In setups like online advertising and cloud computing there is no particular incentive for adversaries to collude, due to conflicting monetary interests. In crowdsourcing scenarios the probability that adversaries who may collude will be assigned to the same task is minuscule due to the large number of anonymous available workers. Attempts to collude can often be detected easily, and the possibility of strict penalization (by the crowdsourcing market) provides additional disincentive to collude. Finally, we note that, an assumption of no collusion is a necessary and a practical one in most of these situations; otherwise there would be no way to accomplish those tasks.

The main contributions of this paper are as follows:
\squishlist
 \item We introduce the problem of {\em privacy-aware data partitioning} across multiple adversaries, and analyze its complexity. To our knowledge this is the first work that addresses the problem of minimizing information leakage when partitioning a dataset across multiple adversaries. 
 \item We introduce \sparsi, a rigorous framework based on the notion of {\em sensitive properties} that allows us to formally reason about how information is leaked and the total amount of information disclosure. We represent the interdependencies between data and sensitive properties using a {\em hypergraph} and we show how the problem of privacy-aware partitioning can be cast as an optimization problem that it is NP-hard by reducing it to hypergraph partitioning. 
 \item We analyze the problem for specific families of information disclosure functions, including step and linear functions, and show how good solutions can be derived by using {\em relaxation} techniques. Furthermore, we propose a set of algorithms, based on a generic {\em greedy randomized local search} algorithm, for obtaining approximate solutions to this problem under generic families of utility and information disclosure functions. 
 \item Finally, we demonstrate how, using SPARSI, one can distribute user-location data, like in Example \ref{ex:socialnetwork}, to multiple advertisers while ensuring that almost no sensitive information about potential user friendship links is revealed. Moreover, we experimentally verify the performance of the proposed algorithms for both synthetic and real-world datasets. We compare the performance of the proposed greedy local search algorithm against approaches tailored to specific disclosure functions, and show that it is capable of producing solutions that are close to the optimal. 
\squishend

\section{SPARSI Framework}
\label{sec:prelim}
In this section we start by describing the different components of \sparsi. More precisely, we show how one can formally reason about the sensitive information contained in a dataset by introducing the notion of {\em sensitive properties}. Then, we show how to model the interdependencies between data entries and sensitive properties rigorously, and how to reason about the leakage of sensitive information in a principled manner.

\subsection{Data Entries and Sensitive Information}
\label{sec:dependencygraph}
Let $D$ denote the dataset to be partitioned among different adversaries. Moreover, let $A$ denote the set of adversaries. 
We assume that $D$ is comprised of data entries $d_i \in D$ that disclose minimal sensitive information if revealed alone. To clarify this consider Example \ref{ex:socialnetwork} where each data entry is the check-in location of a user. The user is sharing this information voluntarily with the social network service in exchange for local advertisement services, hence, this entry is assumed not to disclose sensitive information. In Example \ref{ex:crowdsourcing}, the data entries to be published are the fields of the prescriptions. Observe that if the disease field is revealed in isolation, no information is leaked about possible individuals carrying it.

However, revealing several data entries together discloses sensitive information. We define a {\em sensitive property} to be a property that is related to a subset of data entries but not explicitly represented in the data set, and that can be inferred if the data entries are collectively analyzed. Let $P$ denote the set of sensitive properties that are related to data entries in $D$. To formalize this abstract notion of indirect information disclosure, we assume that each sensitive property
$p \in P$ is associated with a variable (either numerical or categorical) $V_p$ with true value $v^*_p$. Let $D_p \subset D$ be the smallest set of data entries from which an adversary can infer the true value $v^*_p$ of $V_p$ with high confidence, if all entries in $D_p$ are revealed to her. We assume that there is a unique such $D_p$ corresponding to each property $p$. We say that data entries in $D_p$ disclose information about property $p \in P$ and that information disclosure can be modeled as a function over $D_p$ (see Section \ref{sec:infodisclosure}).

We assume that sensitive properties are specified by an expert and the dependencies between data entries in $D$ and properties in $P$, via sets $D_p,~\forall p \in P$, are represented as an undirected bipartite graph, called a {\em dependency graph}. Returning to the example applications presented above we have the following: In Example \ref{ex:socialnetwork} the sensitive properties correspond to the friendship links between users, and the assosiated datasets $D_p$ correspond to the check-in information of the pairs of users participating in friendship links. In Example \ref{ex:crowdsourcing}, the sensitive properties correspond to the links between a patient's id and a particular disease, or a doctor's id and particular medication. In general, it has been shown that data mining techniques can be used to determine the dependencies between data items and sensitive properties~\cite{li08:icde}.

Let $\mathcal{G}_d$ denote such a dependency graph. $\mathcal{G}_d$ has two types of nodes, i.e., nodes $P$ that correspond to sensitive properties and nodes $D$ that correspond to data entries. An edge connects a data entry $d \in D$ with a property $p \in P$ only if $d$ can potentially disclose information about $p$. Alternatively, we can use an equivalent {\em hypergraph} representation, that is easier to reason about in some cases. Converting the dependency graph $\mathcal{G}_d$ into an equivalent {\em dependency hypergraph} is simply done by mapping each property node into a hyperedge. We assume that the dimension (i.e., size of largest hyperedge) of the dependency hypergraph is bigger than the number of adversaries.

\begin{figure}
\begin{center}
\includegraphics[scale=0.3]{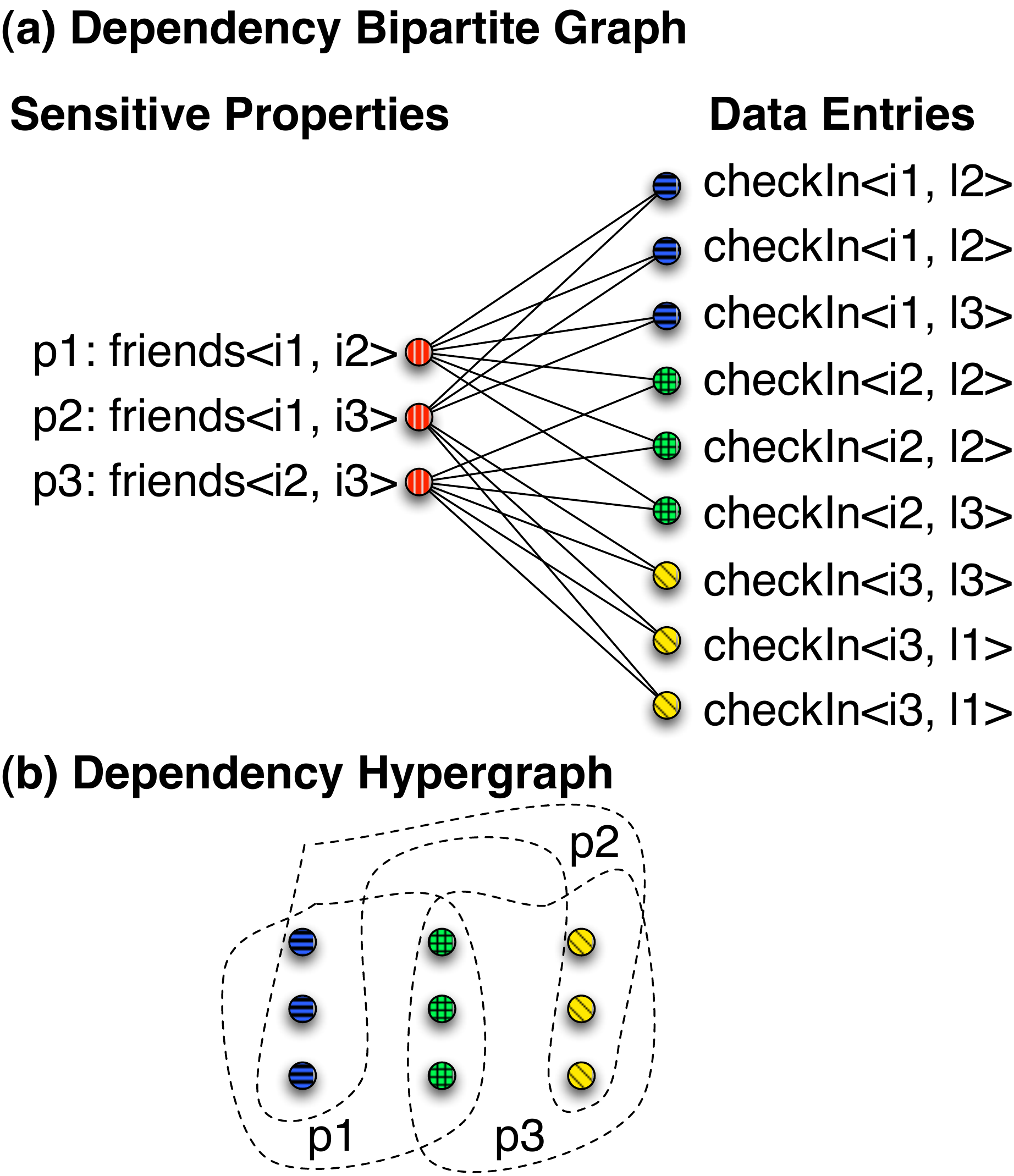}
\caption{An example of a dependency graph between data entries and sensitive properties. Data entries corresponding to the same user are colored using the same color.}
\label{fig:prop_example}
\end{center}
\end{figure}

An example of a bipartite graph and its equivalent hypergraph is shown in Figure \ref{fig:prop_example}. Recall that in this example we do not want to disclose any information about the structure of the social network, i.e., the sensitive properties are the friendship links between individuals. However, if an adversary is given the check-in locations of two individuals, she can infer whether there is a friendship link or not between them \cite{cho:2011}. The dependencies between  check-ins and friendship links are captured by the edges in the bipartite graph.

\subsection{Information Disclosure}
\label{sec:infodisclosure}
We model the information disclosed to an adversary $a$ using a vector valued function $f_a: \mathcal{P}(D) \rightarrow [0,1]^{|P|}$, which takes as input the subset of data entries published to an adversary, and returns a vector of disclosure values; one per sensitive property. That is , $f_a(S_a)[i]$ denotes the information disclosued to adversary $a \in A$ about the $i$th property when $a$ has access to the subset $S_a$ of data entries. We assume that information disclosure takes values in $[0,1]$, with $0$ indicating no disclosure and $1$ indicating full disclosure. Generic disclosure functions, including posterior beliefs, and distribution distances can be naturally represented by \sparsi. The only requirement is that the disclosure function returns a value for each sensitive property. 

Based on the disclosure functions of all adversaries we define the overall disclosure function $f$ as an aggregate  of all functions in $F$. Before presenting the formal definition, we define the {\em assignment set}, given as input to $f$.
\begin{definition}[Assignment Set]
\label{def:aset}
Let $x_{da}$ be an indicator variable set to 1 if data entry $d \in D$ is published to adversary $a \in A$. We define the assignment set $\aset$ to be the set of all variables $x_{da}$, i.e., $\aset = \{x_{11}, \cdots, x_{1|A|}, \cdots, x_{|D||A|} \}$, and the adversary's assignment set $\aset_a$ to be the set of indicator variables corresponding to adversary $a \in A$, i.e.,  $\aset_a = \{x_{1a}, x_{2a},\cdots, x_{|D|a}\}$.
\end{definition}
\newpage
\noindent\textbf{Worst Disclosure.} The overall disclosure can be expressed as:
\begin{equation}
\label{eq:worst_case_disclosure}
f_{\infty}(\aset) = max_{a \in A}( \| f_a(\aset_a) \|_{\infty} )
\end{equation}
Observe that using the infinity norm accounts for the worst case disclosure across properties. Thus, full disclosure for at least one sensitive property suffices to maximize the information leakage. This function is indifferent to the total number of sensitive properties that are fully disclosed in a particular partitioning and gives the same score to all that have at least one fully disclosed property.

However, there are cases where one is not interested in the worst case disclosure but only interested in the total information disclosed to any adversary. Following this observation we introduce another variation of the overall disclosure function that considers the total information disclosure per adversary.

\noindent\textbf{Average Disclosure.} We replace the infinity norm in the equation above with the $L_1$ norm:
 \begin{equation}
 \label{eq:total_disclosure}
  f_{L_1}(\aset) = max_{a \in A}( \frac{\| f_a(\aset_a) \|_{1}}{|P|} )
\end{equation}

Observe that both Equation \ref{eq:worst_case_disclosure} and Equation \ref{eq:total_disclosure} consider the maximum over the disclosure across adversaries, i.e., they can be written as:
\begin{equation}
\label{eq:unified_disclosure}
f(\aset) = max_{a \in A}f ^{\prime}_{a}(\aset_a)
\end{equation}
where $f^{\prime}_a(\aset_a) =  \|f_a(\aset_a) \|_{\infty} $ or $f^{\prime}_a(\aset_a) =   \frac{\| f_a(\aset_a) \|_{1}}{|P|} $.

\subsection{Overall Utility}
\label{sec:utility}
Let $u$ denote the utility derived by partitioning the dataset across multiple adversaries. We have that $u: \mathcal{P}(D \times A) \rightarrow \mathbb{R}^+$, where $\mathcal{P}(D \times A)$ denotes the powerset of possible data-to-adversary assignments. As we show below, this function either quantifies the utility of the adversaries when acquiring part of the dataset $D$ (see Example \ref{ex:socialnetwork}) or the publisher's utility derived by fulfilling a particular task that requires partitioning the data (see Example \ref{ex:crowdsourcing}). Under many real world examples these two different kinds of utility can be unified under a single utility. Consider Example \ref{ex:socialnetwork}. Typically, advertisers pay higher amounts for data that provide higher utility. Thus, maximizing the utility of each individual advertiser maximizes the utility (maybe monetary) of the data publisher as well. 

Based on this observation we unify the two types of utilities under a single formulation based on the utility of adversaries. First we focus on the utility of adversaries. Intuitively, we would expect that the more data an adversary receives, the less the observation of a new, previously unobserved, data entry would increase the gain of the adversaries. This notion of {\em diminishing returns} is formalized by the combinatorial notion of {\em submodularity} and is shown to hold in many real-world scenarios~\cite{marshall:1890,krause:2008}. More formally, a set function $G: 2^V \rightarrow \mathbb{R}$ mapping subsets $A \subseteq V$ into the real numbers is called {\em submodular}~\cite{doerr:2010}, if for all $A \subseteq B \subseteq V$, and $v^{\prime} \in V \setminus B$, it holds that $G(A \cup \{v^{\prime}\}) - G(A) \geq G(B \cup \{v^{\prime}\}) - G(B)$, i.e., adding $v^{\prime}$ to a set $A$ increases $G$ more than adding $v^{\prime}$ to a superset $B$ of $A$. $F$ is called {\em nondecreasing}, if for all $A \subseteq B \subseteq V$ it holds that $G(A) \leq G(B)$.

Let $u_a$ be a set function that quantifies the utility of each adversary $a$. As mentioned above, we assume that $u_a$ is a nondecreasing submodular function. For convenience we will occasionally drop the nondecreasing qualification in the remainder of the paper. Let $U_A$ denote the set of all utility functions for a given set of adversaries $A$. The overall utility function $u$ can be defined as an aggregate function of all utilities $u_a \in U_A$. We require that $u$ is also a submodular function. For example the overall utility may be defined as a linear combination, i.e., a weighted sum, of all functions in $U_A$, following the form:
\begin{equation}
\label{eq:adv_based_util}
u(\aset) = \sum\nolimits_{a \in A} w_a u_a(\aset_a)
\end{equation}
where $\aset$ and $\aset_a$ are defined in Definition \ref{def:aset}.
Since all functions in $U_A$ are submodular, $u$ will also be submodular, since it is expressed as a linear combination of submodular functions~\cite{fujishige:2005}.

An example of a submodular function $u_a$ is an additive function. Assume that each data entry in $d \in D$ has some utility $w_{da}$ for an adversary $a \in A$. We have that $u_a(\aset_a) = \sum\nolimits_{d \in D}w_{da}x_{da}$, where $x_{da}$ is an indicator variable that takes value 1 when data entry $d$ is revealed to adversary $a$ and 0 otherwise. For the remainder of the paper we will assume that utility $u$ is normalize so that $u \in [0,1]$.



\section{Privacy-Aware Data Partitioning}
\label{sec:problem}
In this section, we describe two formulations of the {\em privacy-aware partitioning} problem. We show how both can be expressed as maximization problems that are, in general, NP-hard to solve. We consider a dataset $D$ that needs to be partitioned across a given set of adversaries $A$. We assume that each data entry must be revealed to at least one and at most $t$ adversaries. The lower bound arises naturally in both application discussed in the introduction section. The upper bound is necessary to model cases where the number of assignments per data entry needs to be restricted as it might incur some cost, e.g., monetary in crowdsourcing applications.

We assume that the functions to compute the overall utility and information disclosure are given. Let these functions be denoted by $u$ and $f$ respectively. Ideally, we wish to maximize the utility while minimizing the cost; however, there is a natural tradeoff between the two optimization goals. A traditional approach is to set a requirement on information disclosure while optimizing the gain. Accordingly we can define the following optimization problem.

\begin{definition}[DiscBudget]
Let $D$ be a set of data entries, $A$ be a set of adversaries, and $\tau_{I}$ be a budget on information disclosure.
This formulation of the privacy-aware partitioning problem finds a data entry to adversary assignment set $\aset$ that maximizes $u(\aset)$ under constraint $f(\aset) < \tau_{I}$. More formally we have the following optimization problem:
\begin{equation}
\label{eq:sec_part_max_utility}
\begin{aligned}
& \underset{\aset \in \mathcal{P}(D \times A)}{\text{maximize}}
& & u(\aset) \\
& \text{subject to}
&& f(\aset) < \tau_{I}, \\
&&&1 \leq \sum\nolimits_{a=1}^k x_{da} \leq t, \forall d \in D, \\
&&&x_{da} \in \{0,1\}.
\end{aligned}
\end{equation}
where $x_{da}$ and $\aset$ are defined in Definition \ref{def:aset} as before and $t \geq 1$ is the maximum number of adversaries to whom a particular data entry can be published.
\end{definition}
This optimization problem already captures our desire to reduce the information disclosure while increasing the utility. However, depending on the value of $\tau_{I}$, the optimization problem presented above might be infeasible. Infeasibility occurs when $\tau_{I}$ is so small that no assignment of data to adversaries, such that $\sum\nolimits_{a=1}^k x_{da} \geq 1,~\forall d \in D$ and $f(\aset) < \tau_{I}$ exists.

To overcome this, we consider a different formulation of the privacy-aware data partitioning problem where we seek to maximize the difference between the utility and the information disclosure functions. We consider the Lagrangian relaxation of the previous optimization problem. Again, we assume that both functions are measured using the same unit. We have the following:
\begin{definition}[Tradeoff]
Let $D$ be a set of data entries, $A$ be a set of adversaries, and $\tau_{I}$ be a budget on information disclosure.
This formulation of the privacy-aware partitioning problem finds a data entry to adversary assignment $\aset$ that maximizes the tradeoff between the overall utility and the overall information disclosure, i.e., $u(\aset) + \lambda(\tau_{I}- f(\aset))$, where $\lambda$ is a nonnegative weight. More formally we have the following optimization problem:
\begin{equation}
\label{eq:sec_part_tradeoff}
\begin{aligned}
& \underset{\aset \in \mathcal{P}(D\times|A|)}{\text{maximize}}
&& u(\aset) +\lambda(\tau_{I} - f(\aset)) \\
& \text{subject to}
&& 1 \leq \sum\nolimits_{a=1}^k x_{da} \leq t, \forall d \in D,\\
&&&x_{da} \in \{0,1\}.
\end{aligned}
\end{equation}
where $x_{da}$ and $\aset$ are defined in Def. \ref{def:aset} and $t$ is the maximum number of adversaries to whom a data entry can be published.
\end{definition}
We prove that both problems above are NP-hard by reducing both versions of the privacy-aware data partitioning problem to the {\em hypegraph coloring problem}~\cite{Dinur:2002}.
\begin{theorem}
Both formulations  of the privacy-aware data partitioning problem are, in general, NP-hard to solve.
\end{theorem}
\begin{proof}
Fix a set of adversaries denoted by $A$ and a set of data entries denoted by $D$. Let $P$ denote the sensitive properties that correspond to data entries in $D$. We require to partition $D$ across the adversaries in $A$. Consider the following instance of the privacy-aware data partitioning problem. We require that each data entry be published to exactly one adversary. Moreover, we set the maximum budget on information disclosure to be 1. We also fix the overall information disclosure to be a step function of the following form: If all the data entries corresponding to a particular sensitive property are revealed to the same adversary the overall disclosure is 1 otherwise it is 0.  Finally we consider a constant utility function, which is always equal to 1. Considering the hypergraph representation of data and properties, it is easy to see that this problem is equivalent to the hypergraph coloring problem, which is NP-hard~\cite{Dinur:2002}. Reversing the above steps, one can easily reduce any instance of hypergraph coloring to the privacy-aware data partitioning problem.
\end{proof}

In the remainder of the paper we describe efficient heuristics for solving the partitioning problem -- we present approximation algorithms for specific information disclosure functions in Section \ref{sec:special}, and a greedy local-search heuristic for the general problem in Section \ref{sec:grasp}. Due to space constraints, henceforth, we will only focus on the \tradeoff formulation. Many of our algorithms also work for the \discbudget formulation (with slight modifications).


\section{Analyzing Specific Cases of Information Disclosure}
\label{sec:special}
In this section, we present approximation algorithms when the information disclosure function takes the following special forms:  1) step functions, 2) linearly increasing functions. The utility function is assumed to be submodular. 


\subsection{Step Functions}
\label{sec:step_functions}
Information disclosure functions that correspond to a step function can model cases when each sensitive property $p \in P$ is either fully disclosed or fully private. A natural application of step functions is the crowdsourcing scenario shown in Example \ref{ex:crowdsourcing}.  When certain fields of a medical transcription, e.g., name together with diagnosis, or gender together with the zip code and birth data, are revealed to an adversary the corresponding sensitive property for the patient is revealed. 

We continue by describing such functions formally. Let $D_p \subset D$ be the set of data entries associated with $p$. Property $p$ is fully disclosed only if $D_p$ is published in its entirety to an adversary. This can be modeled using a set of step functions $f_a \in F$: $f_a(D_a)[p] = 1$, if the set of data items assigned to adversary $a$ contains all the elements $D_p$ associated with property $p$. If $D_p \not\subseteq D_a$, then $f_a(D_a) = 0$. Observe that information disclosure is minimized (and is equal to $0$) when no adversary gets all the elements in $D_p$, for all $p$. For step functions we consider worst case disclosure, since ideally we do not want to fully disclose any property.

Considering \discbudget and \tradeoff separately is not meaningful for step functions. Since disclosure can only take the extreme values $\{0,1\}$, $\tau_I = 1$ in \tradeoff should be set to 0. Hence, full disclosure of a property always penalizes the utility.  Hence, one can reformulate the problem and seek for solutions that maximize the utility function under the constraint that information disclosure is 0, i.e., no property exists such that all its corresponding data entries are published to the same adversary.

Given these families of information disclosure functions and a submodular utility function, both formulations of privacy-aware data partitioning can be represented as an \emph{integer program (IP)}:
\begin{equation}
\label{eq:ilp_stepfunctions}
\begin{aligned}
& \underset{\aset \in \mathcal{P}(D\times|A|)}{\text{maximize}}
& & u(\aset) \\
& \text{subject to}
&& \sum\nolimits_{d \in D_p} x_{da} < |D_p|,\forall p \in P, \forall a \in A, \\
&&&1 \leq \sum\nolimits_{a=1}^k x_{da} \leq t, \forall d \in D,\\
&&&x_{da} \in \{0,1\}.
\end{aligned}
\end{equation}
where $t$ is the maximum number of adversaries to whom a particular data entry can be published.

The first constraint enforces that there is no full disclosure of a sensitive property. The partitioning constraint enforces that a data entry is revealed to at least one but no more than $t$ adversaries. Solving the optimization problem in (\ref{eq:ilp_stepfunctions}) corresponds to maximizing a submodular function under linear constraints. Recall that the utility function is submodular and observe that all constraints in the optimization problem presented above are linear. In fact be viewed as packing constraints.

For additive utility functions ($u = \sum_{a\in A}\sum_{d\in D} w_{da}x_{da}$), Equation \ref{eq:ilp_stepfunctions} becomes an integer linear program, that can be approximately solved in PTIME in two steps. First, one can solve a linear relaxation of Equation \ref{eq:ilp_stepfunctions}, where $x_{da}$ is some fraction in $[0,1]$. The resulting fractional solution can be converted into an integral solution using a {\em rounding strategy}.

The simplest rounding strategy, called {\em randomized rounding} \cite{raghavan:1987}, works as follows: assign data entry $d$ to an adversary $a$ with probability equal to $\hat{x}_{da}$, where $\hat{x}_{da}$ is the fractional solution to the linear relaxation. The value of the objective function achieved by the resulting integral solution is in expectation equal to the optimal value of the objective achieved by the linear relaxation. Moreover, randomized rounding preserves all constraints in expectation. A different kind of rounding, called {\em dependent rounding} \cite{gandhi:2006}, ensures that constraints are satisfied in the integral solution with probability 1. For an overview of different randomized rounding techniques and the quality of the derived solutions for budgeted problems we refer the reader to the work by Doerr et al.~\cite{doerr:2010}

One can solve the general problem with worst-case approximation guarantees by leveraging a recent result on submodular maximization under multiple linear constraints by Kulik et al. \cite{Kulik:2009}.
\newpage
\begin{theorem}
Let the overall utility function $u$ be a nondecreasing submodular function. One can find a feasible solution to the optimization problem in (\ref{eq:ilp_stepfunctions}) with expected approximation ratio of $(1-\epsilon)(1-e^{-1})$, for any $\epsilon > 0$, in polynomial time.
\end{theorem}

\begin{proof}
This holds directly by Theorem 2.1 of Kulik et al. \cite{Kulik:2009}. 
\end{proof}

To achieve this approximation ratio, Kulik et al. introduce a framework that first obtains an approximate solution for a continuous relaxation of the problem, and then uses a non-trivial combination of a randomized rounding procedure with two enumeration phases, one on the most {\em profitable elements}, and the other on the `big' elements, i.e., elements with high cost. This combination enables one to show that the rounded solution can be converted to a feasible one with high expected profit. Due to the intricacy of the algorithm we refer the reader to Kulik et al.~\cite{Kulik:2009} for a detailed description of the algorithm.

\subsection{Linearly Increasing Functions}
\label{sec:linear_incr_functions}
In this section, we consider linearly increasing disclosure functions. Linear disclosure functions can naturally model situtations where each data entry independently affects the likelihood of disclosure. In particular, if normalized log-likelihood is used as a measure of information disclosure, the disclosure function takes the linear form presented below. We consider the following additive form for the disclosure of property $p$: 
\begin{equation}
\label{eq:lineardisclosure}
f_{a}(\cdot)[p] = \sum\nolimits_{d \in D_p} a_{dp} x_{da}
\end{equation}
where $a_{dp}$ is a weight associated with the information that is disclosed about property $p$ when data $d$ is revealed to an adversary. 

We can rewrite the \tradeoff problem statement as:
\begin{equation}
\begin{aligned}
& \underset{\aset \in \mathcal{P}(D \times A)}{\text{maximize}}
&& u(\aset) +\lambda(\tau_{I} - \max_{a \in A}(f^{\prime}_a(\aset_a)))\\
& \text{subject to}
&& 1 \leq \sum\nolimits_{a=1}^k x_{da} \leq t, \forall d \in D,\\
&&&x_{da} \in \{0,1\}.
\end{aligned}
\end{equation}
When the utility function is additive, the above problem is an integer linear program, and hence can be solved by rounding the LP relaxation as explained in the previous section.

However, for general submodular $u(\cdot)$, the objective is not submodular anymore -- the max of the additive information disclosure functions is not necessarily supermodular~\cite{fujishige:2005}. Hence, unlike the case of step functions, we cannot use the result of Kulik et al. \cite{Kulik:2009} to get an efficient approximate solution.

Nevertheless, we can compute approximate solutions in PTIME by considering the following max-min formulation of the problem:
\begin{equation}
\label{eq:minimax_linear}
\begin{aligned}
& \underset{\aset \in \mathcal{P}(D \times A)}{\text{maximize}}
&& \min_{a \in A} (u(\aset) + \lambda(\tau_{I}-f^{\prime}_a(\aset_a)))\\
& \text{subject to}
&& 1 \leq \sum\nolimits_{a=1}^k x_{da} \leq t, \forall d \in D,\\
&&&x_{da} \in \{0,1\}.
\end{aligned}
\end{equation}

Since the overall utility function is a nondecreasing submodular function, and the disclosure for each adversary is additive, the objective now is a max-min of submodular functions. More precisely, for worst-case disclosure (Equation \ref{eq:worst_case_disclosure}), the optimization objective can be rewritten as:
\begin{equation}
\begin{aligned}
& \underset{\aset \in \mathcal{P}(D\times A)}{\text{maximize}}
& & \min_{a \in A, p \in P} (u(\aset) + \lambda(\tau_{I} -f_a(\aset_a)[p]))\\
\end{aligned}
\end{equation}
and, for average disclosure (Equation \ref{eq:total_disclosure}), it can be written as:
\begin{equation}
\begin{aligned}
& \underset{\aset \in \mathcal{P}(D\times A)}{\text{maximize}}
&& \min_{a \in A} (u(\aset) + \lambda(\tau_{I} -\frac{1}{|P|}\sum\nolimits_{p \in P}f_a(\aset_a)[p]))\\
\end{aligned}
\end{equation}
The above max-min problem formulation is closely related to the {\em max-min fair allocation} problem \cite{Golovin:2005} for both types of information disclosure functions.
The main difference between Problem (\ref{eq:minimax_linear}) and the max-min fair allocation problem is that data items may be assigned to multiple advrsaries. In the max-min fair allocation problem a data item is assigned exactly once. If $t =1$ then the two problems are equivalent, and thus, one can provide worst case guaranties on the quality of the approximate solution. The problem of max-min fair allocation was studied by Golovin \cite{Golovin:2005}, Khot and Ponnuswami \cite{Khot:2007}, and Goemans et al. \cite{Goemans:2009}. Let $n$ denote the total number of data entries ({\em goods} in the max-min fair allocation problem) and $m$ denote the number of adversaries ({\em buyers} in the max-min fair allocation problem). The first two papers focus on additive functions and give algorithms achieving an $(n - m + 1)$-approximation and a $(2m + 1)$-approximation respectively, while the third one gives a $\mathcal{O}(n^{\frac{1}{2}}m^{\frac{1}{4}}\log n \log^{\frac{3}{2}}m)$-approximation.

\subsection{Quadratic Functions}
\label{sec:quad}
In this section 
we consider \emph{quadratic} disclosure functions of the following form:
\begin{equation} \label{eq:quad_func}
f_{a}(\cdot)[p] = \left( \sum\nolimits_{d \in D_p} a_{dp}x_{da}\right)^2
\end{equation}
where $a_{dp}$ and $x_{da}$ are defined as before. Since $f_{a}(D_p)[p] = 1$ we have that $(\sum\nolimits_{d \in D_p} a_{dp})^2 = 1$.
We assume that the utility is an additive function following the form of Equation \ref{eq:adv_based_util}, and do not consider generic submodular case.

The \tradeoff optimization function can be rewritten as:
\begin{equation}
\begin{aligned}
& \underset{\aset \in \mathcal{P}(D \times A)}{\text{maximize}}
&& \sum\nolimits_{d \in D} \sum\nolimits_{a \in A} w_{da}x_{da}  + \lambda(\tau_{I} - \max_{a \in A}(f^{\prime}_a(\aset_a)))\\
\end{aligned}
\end{equation}
where $f^{\prime}_a(\cdot)$ is a quadratic function. The internal maximization over information disclosure functions can be removed from the objective by rewriting the optimization problem as:
\begin{equation}
\label{eq:max_trade_quadratic_min}
\begin{aligned}
& \underset{\aset \in \mathcal{P}(D \times A)}{\text{maximize}}
& & \sum\nolimits_{d \in D} \sum\nolimits_{a \in A} w_{da}x_{da}
 + \lambda(\tau_{I}- y) \\
& \text{subject to}
&& y \geq f^{\prime}_a(\aset_a), \forall a \in A, \\
&&& 1 \leq \sum\nolimits_{a=1}^k x_{da} \leq t, \forall d \in D,\\
&&&x_{da} \in \{0,1\}.
\end{aligned}
\end{equation}
Since all constraints are either linear or quadratic, the above problem is an integer {\em 0-1 Quadratic Constraint Problem (QCP)} \cite{Boyd:2004}, which is, in general, NP-hard to solve. In order to derive an approximate solution in PTIME, we relax this problem to an equivalent {\em Second Order Conic Programming (SOCP)} problem \cite{Boyd:2004}. A SOCP problem can be solved in polynomial time to within any level of accuracy by using an {\em interior-point} method~\cite{Boyd:2004} resulting in a fractional solution $\hat{x}_{da}$.

First, we focus on the constraints shown in Problem (\ref{eq:max_trade_quadratic_min}). The constraints for worst and average disclosure can be written as:
\begin{equation}
\label{eq:cons2}
\begin{aligned}
&(\sum\nolimits_{d \in D_p} a_{dp} x_{da})^2 = \mathsf{(A_{ap}x_{ap})^TA_{ap}x_a},~\forall p \in P,~~\forall a \in A \\
&\frac{1}{{|P|}}\sum\nolimits_{p \in P}(\sum\nolimits_{d \in D_p} a_{dp} x_{da})^2 = \mathsf{(A_{p}x_a)^T}\mathsf{A_{p}x_a},~~\forall a \in A \\
\end{aligned}
\end{equation}
where $\mathsf{x_a}$ corresponds to a vector representation of the assignment set $\aset_a$ and $\mathsf{A^T_{p}} = [a^{\prime}_1, a^{\prime}_2, \dots , a^{\prime}_{|D|}]$  is a positive vector that contains the appropriate weights for all data entries with respect to all properties $p \in P$.

Both constraints follow the quadratic form $\mathsf{(Ax)^T}\mathsf{Ax}$, where $\mathsf{x}$ is a vector representation of the assignment set $\aset_a$ and $A$ is a matrix containing the appropriate $A_{ap}$'s or $A_{p}$'s based on the type of information disclosure we are using. Let $C$ denote the total number of constraints with respect to the disclosure function. Observe that we have $C = |P||A|$ and $C = |A|$ for the two cases of information disclosure respectively.

The next step is to incorporate variable $y$ in the optimization problem. For that we extend vector $\mathsf{x}$ to include variable $y$. The new variable vector is $\left[ y~\mathsf{x} \right]^{\mathsf{T}}$. We can rewrite the quantities in Equation \ref{eq:cons2} as follows:
\begin{equation}
\left(\left[ \begin{array}{cc} 0 & \mathsf{A} \end{array} \right] \left[ \begin{array}{c} y \\ \mathsf{x} \end{array} \right]\right)^{\mathsf{T}}\left(\left[ \begin{array}{cc} 0 & \mathsf{A} \end{array} \right] \left[ \begin{array}{c} y \\ \mathsf{x} \end{array} \right]\right),~\forall c \in [C]
\end{equation}
where $\mathsf{A}$ and $\mathsf{x}$ are as defined above. Finally, we have that the equivalent SOCP problem to the initial QCP problem is:
\begin{equation}
\label{eq:max_min_socp}
\begin{aligned}
& \underset{q}{\text{maximize}}
&& \left[ \begin{array}{cc} -\lambda &\mathsf{W} \end{array} \right] \mathsf{q} +  \lambda\tau_{I}\\
& \text{subject to}
&& \left[ \begin{array}{cc}1 & 0\end{array}\right]\mathsf{q} \geq \mathsf{(A^{\prime}q)^T(A^{\prime}q)},~\forall c \in [C], \\
&&& 1 \leq \sum\nolimits_{a=1}^k x_{da} \leq t, \forall d \in D.
\end{aligned}
\end{equation}
where $\mathsf{q} = \left[ \begin{array}{cc} y & \mathsf{x} \end{array} \right]^{\mathsf{T}}$, $\mathsf{A}^{\prime} = \left[ \begin{array}{cc} 0 & \mathsf{A} \end{array} \right]$.

Finally, the fractional solutions $\hat{x}_{da}$ obtained from the SOCP needs to be converted to an integral solution. We point out that for the \tradeoff problem no guarantees can be derived on the value of the objective function of the integral solution, when naive randomized rounding schemes, such as setting $x_{da} = 1$ with $\Pr[X_{ad} = 1] = \hat{x}_{da}$ are used. Thus, finding a rounding scheme that ensures that the objective of the integral solution is equal to that of the fractional solution in expectation is an open problem.


\section{A Greedy Local Search \\Heuristic}
\label{sec:grasp}
So far we studied specific families of disclosure functions to derive worst-case guaranties for the quality of approximate solutions. In this section we present two greedy heuristic algorithms suitable for any general disclosure function. We still require the utility function to be submodular. Our heuristics are based on hill climbing and the {\em Greedy Randomized Adaptive Local Search Procedure}(GRASP)~\cite{feo:1995}. Notice that local search heuristics are known to perform well when maximizing a submodular objective function \cite{fujishige:2005}. 
Again, we only focus on the \tradeoff optimization problem(see ~Equation \ref{eq:sec_part_tradeoff}).

{
\small{
\begin{algorithm}[h]
\caption{Overall Algorithm}
\begin{algorithmic}[1]
\STATE {\bf Input:} $A$: set of adversaries; $G$: objective function;  \\
                $r$: number of repetitions; $t$: max. adversaries per data item
\STATE {\bf Output:} $M_{opt}$: a data-to-adversary assignment matrix
\FORALL{$i = 1 \to r$}
    \STATE $M_{\emptyset} \leftarrow\mbox{ empty assignment }$, $g_{opt} \leftarrow G(M_{\emptyset})$
	\STATE $\langle M_{ini},g_{ini} \rangle \leftarrow {\sf CONSTRUCTION}(G,A,t)$;
	\STATE $\langle M,g \rangle \leftarrow {\sf LOCALSEARCH}(G,A,t,M_{ini},g_{ini})$;
	\IF {$g > g_{opt}$}
		\STATE $M_{opt} \leftarrow M$; $g_{opt} \leftarrow g$;
	\ENDIF
\ENDFOR
\RETURN $M_{opt}$;
\end{algorithmic}
\label{algo:grasp}
\end{algorithm}
}
}
\subsection{Overall Algorithm}
Our algorithm proceeds in two phases (Algorithm~\ref{algo:grasp}). The first phase, which we call {\em construction}, constructs a data-to-adversary assignment matrix $M_{ini}$ by greedily picking assignments that maximize the specified objective function $G(\cdot)$, i.e., the tradeoff between utility and information disclosure, while ensuring that each data item is assigned to at least one and at most $t$ adversaries. The second phase, called {\em local-search}, searches for a better solution in the neighborhood of the $M_{ini}$, by changing one assignment of one data item at a time if it improves the objective function, resulting in an assignment $M$. The construction algorithm may be randomized; hence, the overall algorithm is executed $r$ times, and the best solution $M_{opt} = \mbox{argmax}_{\{M_1, \ldots, M_r\}} G(M_i)$ is returned as the final solution.

{
\small{
\begin{algorithm}[h]
\caption{CONSTRUCTION}
\begin{algorithmic}[1]
\STATE {\bf Input:} $G$: objective function; $A$: set of adversaries; \\
$t$: max. adversaries per data item\\
\STATE {\bf Output:} $\langle M, g \rangle$: data-to-adversary assignment, objective value
\STATE $maxIterations \leftarrow t\cdot|D|$
\STATE Initialize: $M \leftarrow \mbox{ empty assignment }$
\FORALL {$i \in [1,maxIterations]$}
    \STATE {\em // Compute a set of candidate assignments}
	\STATE $D_M \leftarrow \mbox{ data entries assigned to } <t \mbox{ adversaries in } M $;
    \STATE Let $S \leftarrow D_M \times A - M$
    \STATE {\em // Pick the next best assignment that improves the objective}
	\STATE $\langle d,a \rangle \leftarrow \mbox{ {\sc PickNextBest}}(M, g, S, G)$
	\IF {$\langle d,a \rangle$ is NULL}
            \STATE break; {\em // No new assignments improve the objective}
	\ENDIF
	\STATE Assign the selected data entry $d$ to the selected adversary $a$;
\ENDFOR
\RETURN $\langle M, G(M) \rangle$;
\end{algorithmic}
\label{algo:grasp_construction}
\end{algorithm}
}}

\subsection{Construction Phase}
The construction phase (Algorithm \ref{algo:grasp_construction}) starts with an empty data-to-adversary assignment matrix and greedily adds a new $\langle d,a \rangle$ assignment to the mapping $M$ if it improves the objective function. This is achieved by iteratively performing two steps. The algorithm first computes a set of candidate assignments $S$. For any data item $d$ (which does not already have $t$ assignments), and any adversary $a$, $\langle d,a \rangle$ is a candidate assignment if it does not appear in $M$.


\eat{
\small{
\begin{algorithm}[h]
\caption{{\sc SelectCandidates}}
\begin{algorithmic}[1]
\STATE {\bf Input:} $A$: set of adversaries; $M$: current assignment; \\
$t$: max adversaries per data item\\
\STATE {\bf Output:} $S$: Set of candidate assignments
\STATE {\bf GLOBAL:}
\STATE \ \ \ $D_M \leftarrow \mbox{ data items assigned to } <t \mbox{ adversaries in } M $;
\STATE \ \ \ $S \leftarrow D_M \times A - M$
\STATE {\bf LOCAL:}
\STATE \ \ \ $d_{cur} \leftarrow \mbox{ data item considered in the previous iteration}$
\STATE \ \ \ $D_M \leftarrow \mbox{ set of data items with a larger id}$;
\STATE \ \ \ If $d_{cur}$ is assigned to $<t$ adversaries, $D_M \leftarrow D_M \cup \{d_{cur}\}$
\STATE \ \ \ $S \leftarrow D_M \times A - M$
\STATE
\RETURN $S$
\end{algorithmic}
\label{algo:select}
\end{algorithm}
}}

{
\small{
\begin{algorithm}[h]
\caption{{\sc PickNextBest}}
\begin{algorithmic}[1]
\STATE {\bf Input:} $G$: objective function; $M$: current assignment; \\
$g$: current value of objective, $S$: possible new assignments\\
\STATE {\bf Output:} new assignment $\langle d^\star,a^\star \rangle$, or NULL
\STATE {\bf GREEDY:}
\STATE \ \ \ $\langle d^\star,a^\star \rangle \leftarrow \mbox{argmax}_{\langle d,a \rangle \in S}G(M \cup \langle d,a \rangle)$
\STATE {\bf GRASP:}
\STATE \ \ \ Pick the top-$n$ assignments $S_n$ having the highest value for $g_{\langle d,a \rangle} = G(M \cup \langle d,a \rangle)$ from $S$, and having $g_{\langle d,a \rangle} > g$.
\STATE \ \ \ $\langle d^\star,a^\star \rangle$ is drawn uniformly at random from $S_n$
\IF {$G(M \cup \langle d,a \rangle) > g$}
\RETURN $\langle d,a \rangle$
\ELSE
\RETURN NULL
\ENDIF
\end{algorithmic}
\label{algo:picknext}
\end{algorithm}
}}

Second, the algorithm picks the next best assignment from the candidates (using {\sc PickNextBest}, Algorithm~\ref{algo:picknext}). We consider two methods for picking the next best assignment -- GREEDY and GRASP. The GREEDY strategy picks $\langle d^\star, a^\star \rangle$ that maximizes the objective $G(M \cup \langle d^\star, a^\star \rangle)$. On the other hand, GRASP identifies a set $S_n$ of top $n$ assignments that have the highest value for the objective  $g_{\langle d, a \rangle} = G(M \cup \langle d, a \rangle)$, such that $g_{\langle d, a \rangle}$ is greater than the current value of the objective $g$. Note that $S_n$ can contain less than $n$ assignments. The GRASP strategy picks an assignment $\langle d^\star, a^\star \rangle$ at random from $S_n$. Both strategies return NULL if $\langle d^\star, a^\star \rangle$ does not improve the value of the objective function.
The construction stops when no new assignment can improve the objective function.

\noindent{\bf Complexity:}
The run time complexity of the construction phase is $O(t\cdot |A|\cdot |D|^2)$. There are $O(t \cdot|D|)$ iterations, and each iteration may have a worst case running time of $O(|D| \cdot |A|)$. {\sc PickNextBest} permits a simple parallel implementation.


{\small{
\begin{algorithm}[h]
\caption{LOCALSEARCH}
\begin{algorithmic}[1]
\STATE {\bf Input:} $G$: objective function; $A$: set of adversaries;\\
 $t$: max. assignments per data item; $M$: current assignment; \\
 $g$: current objective value
\STATE {\bf Output:} $\langle M_{opt},g_{opt}\rangle$: the new assignment, the corresponding objective value
\FORALL {$d \in D$}
	\STATE $A_d \leftarrow $ the set of adversaries to whom  data item $d$ is assigned (according to current assignment $M$);
    \STATE {\em // Construct a set of neighboring assignments}
    \STATE $N_d \leftarrow \{M\}$.
    \STATE  {\bf if} ($|A_d| < t$) {\bf then} $N_d \leftarrow N_d \cup \{ M \cup \{\langle d, a'\rangle \}| \forall a' \not\in A_d \}$;
	\FOR {each adversary $a \in A_d$}
		\STATE  $N_d \leftarrow N_d \cup \{M - \{\langle d,a \rangle\}\}$
		\STATE  $N_d \leftarrow N_d \cup \{ M - \{\langle d,a \rangle\} \cup \{\langle d, a'\rangle \}| \forall a' \not\in A_d \}$;
	\ENDFOR
    \STATE {\em // Pick the neighboring assignment with maximum objective}
    \STATE $M_{opt} \leftarrow \mbox{argmax}_{M' \in N_d} G(M')$
    \STATE $M \leftarrow M_{opt}$
\ENDFOR
\RETURN $\langle M_{opt}, G(M_{opt})\rangle$;
\end{algorithmic}
\label{algo:grasp_localsearch}
\end{algorithm}
}
}

\subsection{Local-Search Phase}
The second phase employs local search (Algorithm \ref{algo:grasp_localsearch}) to improve the initial assignment $M_{ini}$ output by the construction phase. In this phase, the data items are considered exactly once in some (random) order. Given the current assignment $M$, for each data item, a set of neighboring assignments $N_d$ (including $M$) are considered by (i) removing an assignment to an adversary $a$ in $M$, (ii) modifying the assignment from adversary $a$ to an adversary $a'$ (that $d$ was not already assigned to in $M$), and (iii) adding a new assignment (if $d$ is not already assigned to $t$ adversaries in $M$). Next, the neighboring assignment in $N_d$ with the maximum value for the objective $M_{opt}$ is picked. The next iteration considers the data item succeeding $d$ (in the ordering) with $M_{opt}$ as the current assignment. We found that making a second pass of the dataset in the local search phase does not improve the value of the objective function.

\noindent{\bf Complexity:} The run time complexity of the local-search phase is $O(t\cdot |A|\cdot |D|)$.

\subsection{Extensions}
\label{sec:extensions}
The construction phase (Algorithm \ref{algo:grasp_construction}) has a run time that is quadratic in the size of $D$. This is because in each iteration, the {\sc PickNextBest} subroutine computes a {\em global} maximum assignment across all data-adversary pairs. While this approach makes the algorithm more effective in avoiding local minima it reduces its scalability due to its quadratic cost.

To improve scalability, one can adopt a {\em local} myopic approach during construction. Instead of considering all possible (data,adversary) pairs when constructing the list of candidate assignments (see Ln.~8 in Algorithm~\ref{algo:grasp_construction}), one can consider a single data entry $d$ and populate the set of candidate assignments $S$ using only (data,adversary) pairs that contain $d$. More specifically, we fix a total ordering of the data entries $\mathcal{O}$, and perform $t$ iterations of the following:
\squishlist
\item Consider the next data item $d$ in $\mathcal{O}$. Let $M$ be the current assignment.
\item Construct $S$ as $(\{d\}\times A) - M$.
\item Pick the next best assignment in $S$ using Algorithm \ref{algo:picknext} (GREEDY or GRASP) that improves the objective function.
\item Update the current assignment $M$, and proceed with the next data entry in $\mathcal{O}$.
\squishend
\noindent{\bf Complexity:} The run time complexity of the myopic-construction phase is $O(t\cdot |A|\cdot |D|)$.

\subsection{Correctness}
\label{sec:correctness}
While both the construction and local search phases carefully ensure that each data item is assigned to no more that $t$ adversaries, we still need to prove that each data item is assigned to at least one adversary. To ensure this lower bound on the cardinality, we use the following objective function $G$:
\begin{equation}
\label{eq:grasp_construct_objective}
G(\cdot) =  u(\cdot) +\lambda(\tau_{I} - f(\cdot)) - \mathcal{C}
\end{equation}
where $\mathcal{C}$ is the number of data items that are not assigned to any adversary and $ \lambda \in [0,1]$.
The above objective function adds a penalty term $\mathcal{C}$ to the tradeoff between the utility and the information disclosure, i.e., $u(\cdot) +\lambda(\tau_{I} - f(\cdot))$. We can show that this penalty ensures that every data item is assigned to at least one adversary. We have the following theorem:
\begin{theorem}
Using $G(\cdot) =  u(\cdot) +\lambda(\tau_{I} - f(\cdot)) - \mathcal{C}$ with $\lambda \in [0,1]$ as the objective function in Algorithm~\ref{algo:grasp} returns a solution where all cardinality constraints are satisfied.
\end{theorem}
\begin{proof}
Both the construction and the local-search phases ensure that no adversary is assigned to more than $t$ adversaries (Ln.~7 in Algorithms~\ref{algo:grasp_construction} and \ref{algo:grasp_localsearch}).

We need to show that every data item is assigned to at least one adversary, i.e., at the end of the algorithm $\mathcal{C}$ will be equal to 0. We will focus on the global Algorithm~\ref{algo:grasp}. The proof for the local version of the algorithm is analogous.

The main intuition behind the proof is the following:
\squishlist
\item At the end of the construction phase, if $\mathcal{C} > 0$, then some data item must be assigned to $>t$ adversaries.
\item In the local-search phase, making a data item unassigned never results in a better objective function.
\squishend

At the start of the construction phase, $\mathcal{C} = |D|$. We show that if at some iteration $\mathcal{C} = i$, then in that iteration some unassigned data item is assigned to an adversary and $\mathcal{C}$ reduces by 1.
If $\mathcal{C} = i > 0$, there are three possible paths the algorithm can follow after the iteration is over: (1) no new assignment is chosen, (2) the algorithm chooses to assign a data entry to an adversary so that the number of violated constraints remains the same (i.e. that data entry is already assigned to at least one adversary), and (3) a data-to-adversary assignment is chosen so that $\mathcal{C} = i - 1$. We will show that given the objective function we are using, the first path will never be chosen.

We evaluate the value of the objective function for paths 1 and  3. For the third path we consider the worst case scenario, where only disclosure is increased (by $\Delta_f$) and utility remains the same. The value of the objective function for paths 1 and 3 is as follows:
\begin{equation*}
\begin{aligned}
&G_{1} = u + \lambda(\tau_I - f) - \mathcal{C} \\
&G_{3} = u + \lambda(\tau_I - f - \Delta_f) - (\mathcal{C} - 1)\\
\end{aligned}
\end{equation*}
Since $\Delta_f \leq 1$, $G_3 \geq G_1$, thus, path 1 will never be taken if $\mathcal{C} > 0$. 

Notice, that the objective values for paths 2 and 3 cannot be directly compared. However, we will show that during the $t |D|$ iterations we perform path 3 will be chosen $|D|$ times, and, hence $\mathcal{C} = 0$ at the end of Algorithm \ref{algo:grasp_construction}. Let $n_1$, $n_2$ and $n_3$ denote the number of times path1, path 2 and path 3 are chosen respectively. By construction we have that $n_1 + n_2 + n_3 = t |D|$. After $t |D|$ iterations if $\mathcal{C} = 0$ we are done. If $\mathcal{C} > 0$ then based on the above statement we will have $n_1 = 0$, and thus $n_2 + n_3 = t \times |D|$. Now we will show that $n_3$ will always be equal to $|D|$. Let $n_3 < |D|$, we have that $n_2 = t  D - n_3 > (t-1)|D|$. Recall that when path 2 is chosen, the algorithm assigned to an adversary a data entry is already assigned to another adversary. Therefore, if $n_2 > (t-1)|D|$, then some data item is assigned to $>t$ adversaries, which does not happen. Therefore, $n_3 = |D|$.

To prove that $\mathcal{C} = 0$ at the end of the local-search phase, it suffices to show that the algorithm will never choose to remove an assignment from a data item assigned to exactly one adversary. Consider the state of the local-search algorithm right before an iteration of its main for-loop (Ln.~3-13 of Algorithm \ref{algo:grasp_localsearch}). We assume that the algorithm considers a data-entry $d$ with a single assignment and that $\mathcal{C} = 0$. Moreover, let $u$ and $f$ be the utility and disclosure at that point. Let $u^\prime$ and $f^{\prime}$ denote the utility and disclosure if $d$ is assigned to no adversary. Consider the best-case scenario where no utility is lost, i.e., $u^{\prime} = u$, and $f^{\prime} = 0$. Note that $\mathcal{C} = 1$. The new objective value will be $u^{\prime} +\lambda(\tau_I -  f^{\prime})  - \mathcal{C} = u +\lambda \tau_I - 1$. We have that $u +\lambda \tau_I - 1 \leq u + \lambda \tau_I - f$, since $f \leq 1$ always holds.
\end{proof}

\section{Experiments}
\label{sec:experiments}
In this section we present an empirical evaluation of \sparsi. The main questions we seek to address are:
(1) how the two versions of the privacy-aware partitioning problem -- \tradeoff and \discbudget -- compare with each other, and how well they exploit the presence of multiple adversaries with respect to disclosure and utility, (2) how the different algorithms perform in optimizing the overall utility and disclosure, and (3) how practical \sparsi is for distributing real-world datasets across multiple adversaries. 

We empirically study these questions using both real and synthetic datasets. After describing the data and the experimental methodology, we present results that demonstrate the effectiveness of our framework on partitioning sensitive data. The evaluation is performed on an Intel(R) Core(TM) i5 2.3 GHz/64bit/8GB machine. \sparsi is implemented in MATLAB and uses MOSEK, a commercial optimization toolbox.

\textbf{Real-World Dataset:} For the first set of experiments we present how \sparsi can be applied to real world domains. We considered a social networking scenario as discussed in Example \ref{ex:socialnetwork}. We used the {\em Brighkite} dataset published by Cho et al.~\cite{cho:2011}. This dataset was extracted from Brightkite, a former location-based social networking service provider where users shared their locations by checking-in. Each check-in entry contains information about the id of the user, the timestamp and location of the check-in. The dataset contains the public check-in data of users and the friendship network of users. The original dataset contains $4.5$ million check-ins, $58,228$ users and $214,078$ edges. We subsampled the dataset and extracted a dataset comprised of $365,907$ check-ins. The corresponding friendship network contains $3,266$ nodes and $2,935$ edges. In Section \ref{sec:real_data}, we discuss how we modeled the utility and information disclosure for this data.

\textbf{Synthetic Data:} For the second set of experiments we used synthetically generated data to understand the properties of different disclosure functions and the performance of the proposed algorithms better. There are two data-related components in our framework. The first is a hypergraph that describes the interaction between data entries and sensitive properties (see Section \ref{sec:dependencygraph}), and the second is a set of weights $w_{da}$ representing the utility received when data entry $d \in D$ is published to adversary $a \in A$.
The synthetic data are generated as follows. First,  we set the total number of data entries $|D| \in \{50, 100, 200, 300, 500\}$, the total number of sensitive properties $|P| \in \{5,10,50, 100\}$, and the total number of adversaries $|A| \in \{2,3,5,7,10\}$.

Next, we describe the scheme we used to generate the utility weights $w_{da}$.
There are two particular properties that need to be satisfied by such a scheme. The first one is that assigning any entry to an adversary should induce some minimum utility, since it allows us to fulfil the task under consideration (see Example \ref{ex:crowdsourcing}). The second one is that there are cases where certain data items should induce higher utilities when assigned to specific adversaries, e.g., some workers may have better accuracy than others in crowdsourcing, or advertisers may pay more for certain types of data. 

The utility weights need to satisfy the aforementioned properties. To achieve this, we first choose a minimum utility value $u_{\min}$ from a unifrom distribution $\mathcal{U}(0,0.1)$. Then, we iterate over all possible data-to-adversary assignments and set the corresponding weight $w_{da}$ to a value drawn from a uniform distribution $\mathcal{U}(0.8,1)$ with probability $p_{u}$, or to $u_{\min}$ with probability $1-p_u$. For our experiments we set the probability $p_u$ to 0.4. Notice that both properties are satisfied. Finally, weights are scaled down by dividing them with the number of adversaries $|A|$. 

Next, we describe how we generate a random hypergraph $H=(X,E)$, with $|X| = |D|$ and $|E| = |P|$, describing the interaction between data entries and sensitive properties. To create $H$ we simply generate an equivalent bipartite dependency graph $G$ (see Section \ref{sec:dependencygraph}) and convert that to the equivalent dependecy hypergraph. In particular we iterate over the possible data to sensitive property pairs and insert the corresponding edge to $G$ with probability $p_{f}$. For our experiments we set $p_{f}$ to 0.3.

\textbf{Algorithms:} We evaluated the following algorithms:
\squishlist
\item RAND$+$: Each data entry is assigned to exactly $t$ adversaries. The probability of assigning a data entry to an adversary is proportional to the corresponding utility weight $w_{da}$. We run the random partitioning a 100 times, and select the data-to-adversary assignment that maximizes the objective function.
\item LP: We solve the LP relaxation of the optimization problems for step (Section \ref{sec:step_functions}) and linear (Section \ref{sec:linear_incr_functions}) disclosure functions. We generate an integral solution from the resulting fractional solution using naive randomized rounding (see Section \ref{sec:step_functions}). Note that the constraints are satisfied in expectations. Moreover, we perform a second pass over the derived integral solution to guarantee that the cardinality constrains are satisfied. If a data item is not assigned to an adversary, we assign it to the adversary with the highest weight, i.e., corresponding fractional solution. On the other hand is a data item is assigned to more adversaries, we remove those with the lowest weight. This is a naive, yet effective, rounding scheme because the fractional solutions we get are close to the integral ones. More sophisticated rounding techniques can be used~\cite{doerr:2010}. We run the rounding 100 times and select the data-to-adversary assignment with maximum value of objective.
\item ILP: We solve the exact ILP algorithm for step and linear disclosure functions.
\item GREEDY: Algorithm \ref{algo:grasp} with GREEDY strategy for picking a candidate
\item GRASP: Algorithm \ref{algo:grasp} with GRASP strategy for picking the candidate assignments using $n=5$ and $r=10$.
\item GREEDYL: Local myopic variant of Algorithm 1 (see Section \ref{sec:extensions} with GREEDY strategy for picking a candidate).
\item GRASPL: Local myopic variant of Algorithm 1 (see Section \ref{sec:extensions} with GRASP strategy for picking candidates) using $n=\min(k,3)$ and $r=10$.
\squishend

\textbf{Evaluation.} To evaluate the performance of the aforementioned algorithms we used the following metrics: (1) the total utility $u$ corresponding to the final assignment, (2) the information disclosure $f$ for the final assignment and (3) the tradeoff between utility and disclosure, given by $u + \lambda(\tau_I-f)$ . We evaluated the different algorithms using different step and linear information disclosure functions for \tradeoff. 

For all experiments we set $\lambda = 1$ and assume an additive utility function of the form $u_a(\aset_a) = \frac{\sum\nolimits_{d \in D}w_{da}x_{da}}{\sum\nolimits_{d \in D}\sum\nolimits_{{\sf top-t}(A) \in A}w_{da}}$, where $x_{da}$ is an indicator variable that takes value 1 when data entry $d$ is revealed to adversary $a$ and 0 otherwise, and  ${\sf top-t(A)}$ returns the top $t$ adversaries with respect to weights $w_{da}$. Observe that the normalization used corresponds to the maximum total utility a valid data-to-adversary assignment may have, when ignoring disclosure. Using this value ensures that the total utility and the quantity $\tau_I - f$ have the same scale $[0,1]$. For convenience we fixed the upper information disclosure to $\tau_I = 0$. Finally, for $RAND+$ we perform 10 runs and report the average, while for LP we perform the aforementioned rounding procedure 10 times and report the average. The corresponding standard errors are shown as error bars in the plots below.

\subsection{Real Data Experiments}
\label{sec:real_data}
We start by presenting how \sparsi can be applied to real-world applications. In particular, we evaluate the performance of the proposed local-search framework meta-heuristic on generic information disclosure functions using Brightkite. As described in the beginning of the section, this dataset contains the check-in locations of users and their corresponding friendship network. As illustrated in Example \ref{ex:socialnetwork} we desire to distribute the check-in information to advertisers, while minimizing the information we disclose for the structure of the network. We, first, show how \sparsi can be used in this scenario.

\textbf{Utility Weights}. We start by modeling the utility provided when advertisers receive a subset of data entries. As mentioned above each check-in entry contains information about location. We assume a total number of $k$ advertisers, so that each adversary is particularly interested in check-in entries that occurred in a certain geographical area. Given an adversary $a \in A$, we draw $w_{da}$ from a uniform distribution $\mathcal{U}(0.8,1)$ for all entries $d \in D$ that satisfy the location criteria of the adversary, and $w_{da} = 0.1$ otherwise. We simulate this process by performing a random partitioning of the location ids across adversaries. As mentioned above we assume an additive utility function. 

\textbf{Sensitive Properties and Information Disclosure}. The sensitive property in the particular setup is the structure of the social network. More precisely, we require that no information is leaked about the existence of any friendship link among users. It is easy to see that each friendship link is associated with a sensitive property. Now, we examine how check-in entries leak information about the presence of a friendship link. Cho et al.~\cite{cho:2011}  proved that there is a strong correlation between the trajectory similarity and the existence of a friendship link for two users. Computing the trajectory similarity for a pair of users is easy and can be done by computing the cosine similarity of the users given the published set of check-ins. Because of this strong correlation we assume that the information leakage for a sensitive property, i.e., a the link between a pair of users, is equal to the trajectory similarity.

More precisely, let $D_a \subset D$ be the check-in data published to adversary $a \in A$. Let $U$ denote the set of users referenced in $D_a$. Given a sensitive property $p =  e(u_i, u_j),~u_i,u_j \in U,i \neq j$ we have that the information disclosure for $p$ is:
\begin{equation}
f(D_a)[p] = \mathsf{CosineSimilarity}(D_a(u_i),D_a(u_j))
\end{equation}
where $D_a(u_i)$ and $D_a(u_j)$ denote the set of check-in data for users $u_i$ and $u_j$ respectively. We aggregate the given check-ins based on their unique pairs of users and locations, and we extract $15,661$ data entries that contain the user information, the location and the number of times that user visited that particular location. Cosine similarity is computed over these counts new data entries. 

\textbf{Results.} As mentioned above, we aim to minimize the information leaked about any edge in the network. We model this requirement by considering the average case information leakage and setting $\tau_I = 0$. In particular, we would like that no information is leaked at all if possible. Moreover, to partition the dataset we solve the corresponding \tradeoff problem. Since we consider cosine similarity we are limited to using RAND+ and one of the local-search heuristics. In particular, we compare the quality of the solutions for RAND+, GREEDYL and GRASPL. Due to the fact that our implementation of GREEDY and GRASP is single-threaded, these algorithms do not scale for this particular task, and thus, are omitted.  However, as illustrated later (see Section \ref{sec:synth_exp}), these myopic algorithms are very efficient in minimizing the information leakage, thus, suitable for our current objective. Again, we run experiments for $|A| \in \{2,3,5,7,10\}$.

First, we point out that under all experiments GREEDYL and GRASPL reported data-to-adversary assignments with zero disclosure. This means that our algorithms were able to distribute the data entries across adversaries so that {\em no information at all} is leaked about the structure of the social network (i.e., friendship links) to any adversary. On the other hand the {\em average } information disclosure for RAND+ ranges from 0.99 {\em almost full disclosure} to 0.1 as the number of adversaries varies from 2 to 10 (see Table \ref{tab:brightKite_disc}). Disclosing the structure of the entire network with probability almost 1, violates our initial specifications, hence, RAND+ fails to solve the problem when the number of adversaries is small. As the number of adversaries increases, the average amount of disclosed information decreases.

We continue our analysis and focus on the utility and the tradeoff objective. The corresponding results are shown in Figure \ref{fig:brightKite}. Figures \ref{fig:brightKite_util} and \ref{fig:brightKite_obj} correspond to the utility and utility-disclosure tradeoff respectively. The corresponding disclosure is shown in Table \ref{tab:brightKite_disc}. As shown, for a small number of adversaries both GREEDYL and GRASPL generate solutions that induce low values for the total utility. This is expected since both algorithms give particular emphasis to minimizing information disclosure due to the tradeoff formulation of the underlying optimization problem. As the number of adversaries increases both algorithms can exploit the structure of the problem better, and offer solutions that induce utility values that are comparable or higher than the ones reported by RAND+. 

However, looking only at utility values can be misleading, as a very high utility value might also incur a high disclosure value. If fact, RAND+ follows this behavior exactly. The high utility data to adversary assignments, when the number of adversaries is small, is associated with almost full disclosure of the structure of the {\em entire} social network. This is captured in Figure \ref{fig:brightKite_obj}, where we see that local-search algorithms clearly outperform RAND+ since no information is disclosed (see Table \ref{tab:brightKite_disc}). As shown in this figure, in most cases, the average objective value for RAND+ is significantly lower than the ones reported by both GREEDYL and GRASPL. Recall that we compute the average over multiple runs of RAND+, where for each run we execute the algorithm multiple times and consider the best solution reported. The large error bars are indicative of the non-robustness of RAND+ for this problem.
\begin{table}[t]
\small \centering
\caption{Average information disclosure reported by RAND+, GREEDYL and GREEDYL for Brightkite. Notice that local-search algorithms generate solutions that reveal no information about the structure of the friendship network. Standard errors are reported in the parenthesis.}
\begin{tabular}{| c || c || c || c || c || c |}
\hline
\multicolumn{6}{|c|}{Avg. Information Disclosure for different values of k} \\
\hline
$Alg.$ & $k$=2 & $k$=3 & $k$=5 & $k$=7 & $k$=10\\
\hline
RAND+ & 0.99(0) & 0.99(0.3) & 0.3(0.4) & 0.1(0.18) & 0.1(0.018) \\ 
GREEDYL & 0 & 0 & 0 & 0 & 0\\
GRASPL & 0 & 0 & 0 & 0 & 0\\
\hline
\end{tabular}
\label{tab:brightKite_disc}
\end{table}

\begin{figure}[t]
\begin{center}
\subfigure{\includegraphics[trim=50 0 62 0,clip,scale=0.45]{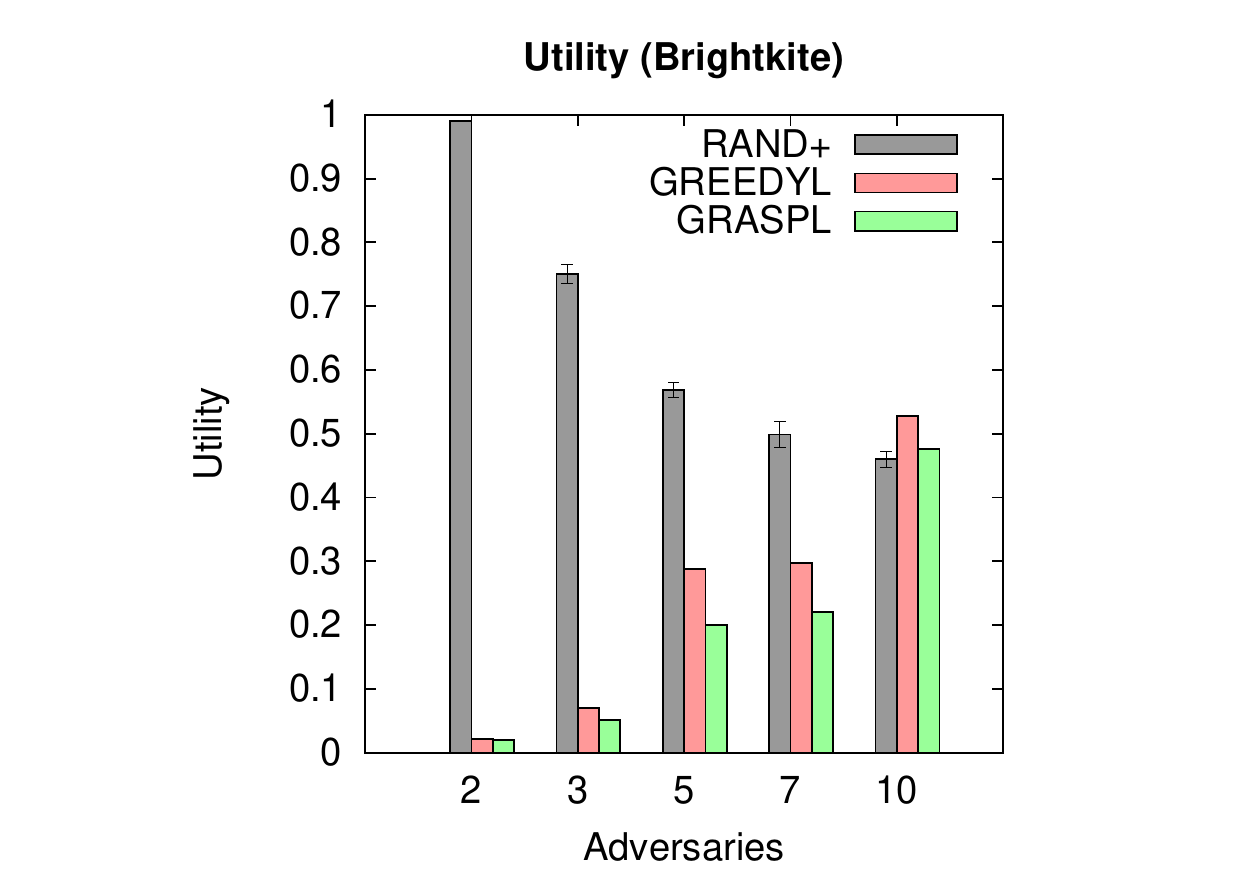} \label{fig:brightKite_util}}
\subfigure{\includegraphics[trim=50 0 62 0,clip,scale=0.45]{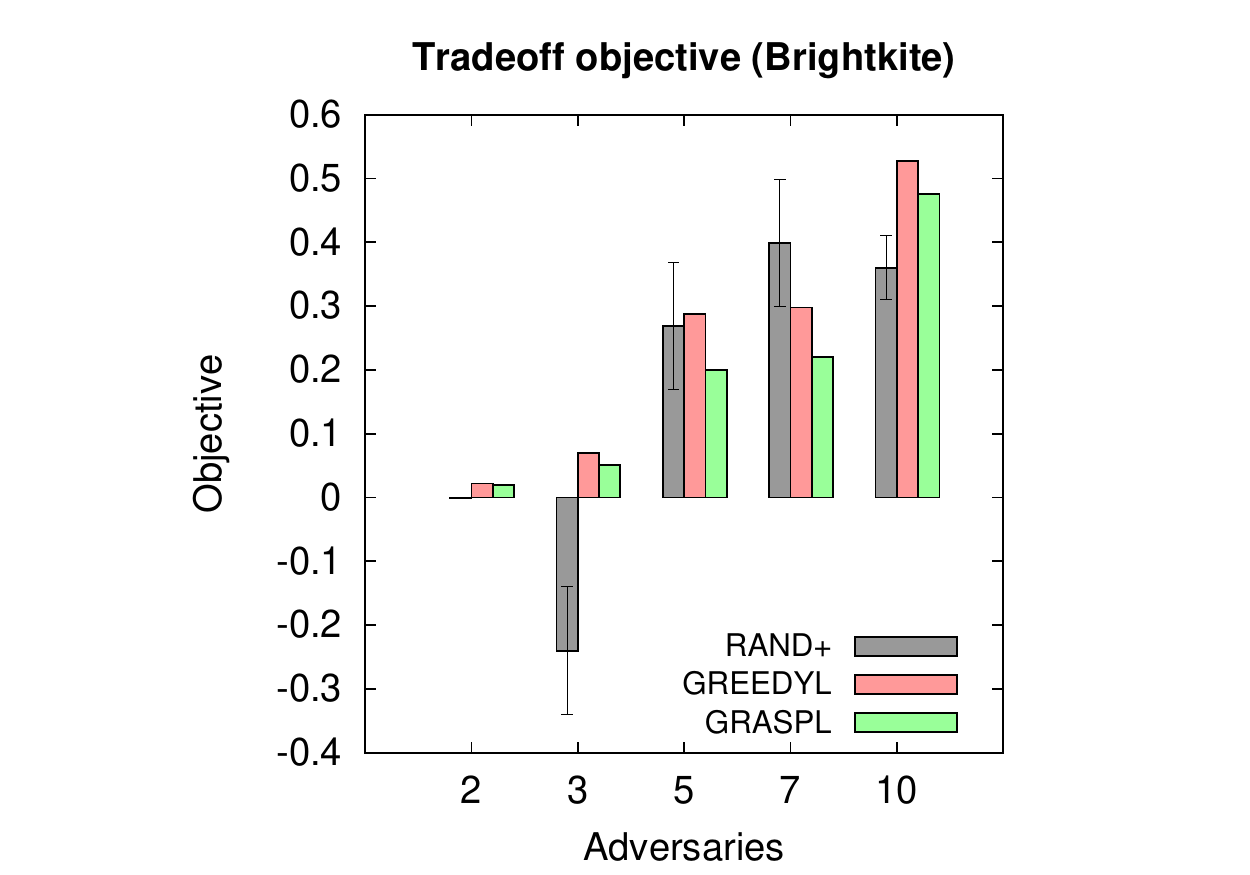} \label{fig:brightKite_obj}}
\end{center}
\caption{Tradeoff objective and utility for  Brightkite. RAND+ generates solutions with high utility but almost full disclosure (see Table \ref{tab:brightKite_disc}). This leads to a poor tradeoff value. GREEDYL and GRASPL disclose no information, and outperform RAND+ with respect to the overall optimization objective.}
\label{fig:brightKite}
\end{figure}

\subsection{Synthetic Data Experiments}
\label{sec:synth_exp}
In this section, we present our results based on the synthetic data. We examined the behavior of the proposed algorithms under several scenarios, where we varied the properties of the dataset to be partitioned, the number of adversaries and the family of disclosure functions considered.

\textbf{Step Functions.} We started by considering step functions. Under this family of disclosure functions, both \discbudget and \tradeoff correspond to the same optimization problem (see Section \ref{sec:step_functions}). Moreover, assuming feasibility of the optimization problem, information disclosure will always be zero. In such cases considering the total utility alone is sufficient to compare the performance of the different algorithms. However, when information is disclosed, comparing the number of fully disclosed properties allows us to evaluate the performance of the different algorithms. 

First, we fixed the number of data entries in the dataset to be $|D| = 500$ and considered values of $|P|$ in $\{5, 10, 50, 100\}$. Figure \ref{fig:step_fixed_data} shows the utility derived by the data-to-adversary assignment corresponding to different algorithms for $|P| = 50$. As depicted, all algorithms that exploit the structure of the dependency graph while solving the underlying optimization problem (i.e., LP, GREEDY, GREEDYL, GRASP and GRASPL) outperform RAND$+$. In most cases, LP, GREEDYL, GREEDY and GRASP where able to find the optimal solution that ILP reported. The high performance of the LP algorithm is justified by the fact that the fractional solution reported was in fact an integral solution. 

GRASPL found solutions with non-optimal utilities, which are still better than RAND$+$. We conjecture that this performance decrease is due to randomization. This is more obvious if we contrast the performance of GRASPL with GRASP. We can see that randomization leads to worse solutions (with respect to utility) when keeping a myopic view on the given optimization problem. On the other hand randomization is helpful in the case of non-myopic local-search. Recall that the reported numbers correspond to no information disclosure, while missing values correspond to full disclosure of at least one sensitive property. As depicted GREEDY failed to find a valid solution when splitting the data across 2 adversaries. However, when randomization was used, GRASP was able to find the optimal solution. Similar performance for different values of $|D|$. These results are omitted due to space constraints.

Next, we ran a second set of experiments to investigate how the performance of the proposed heuristics with respect to the amount of disclosed information. Recall that all local-search algorithms do not explicitly check for infeasible values of disclosure for step functions, since they optimize the tradeoff between utility and disclosure. Therefore, they might report infeasible solutions. We fixed the number of sensitive properties to $|P| = 50$, and the number of adversaries to $k = 2$, and we varied the number of data items $|D|$. We considered $|D| \in \{100, 200, 300, 500\}$. We observed the same behavior as in the previous experiment, i.e., all local-search heuristics failed to report feasible solutions.


\begin{figure}[h]
\begin{center}
\includegraphics[trim=5 0 10 0,clip,scale=0.5]{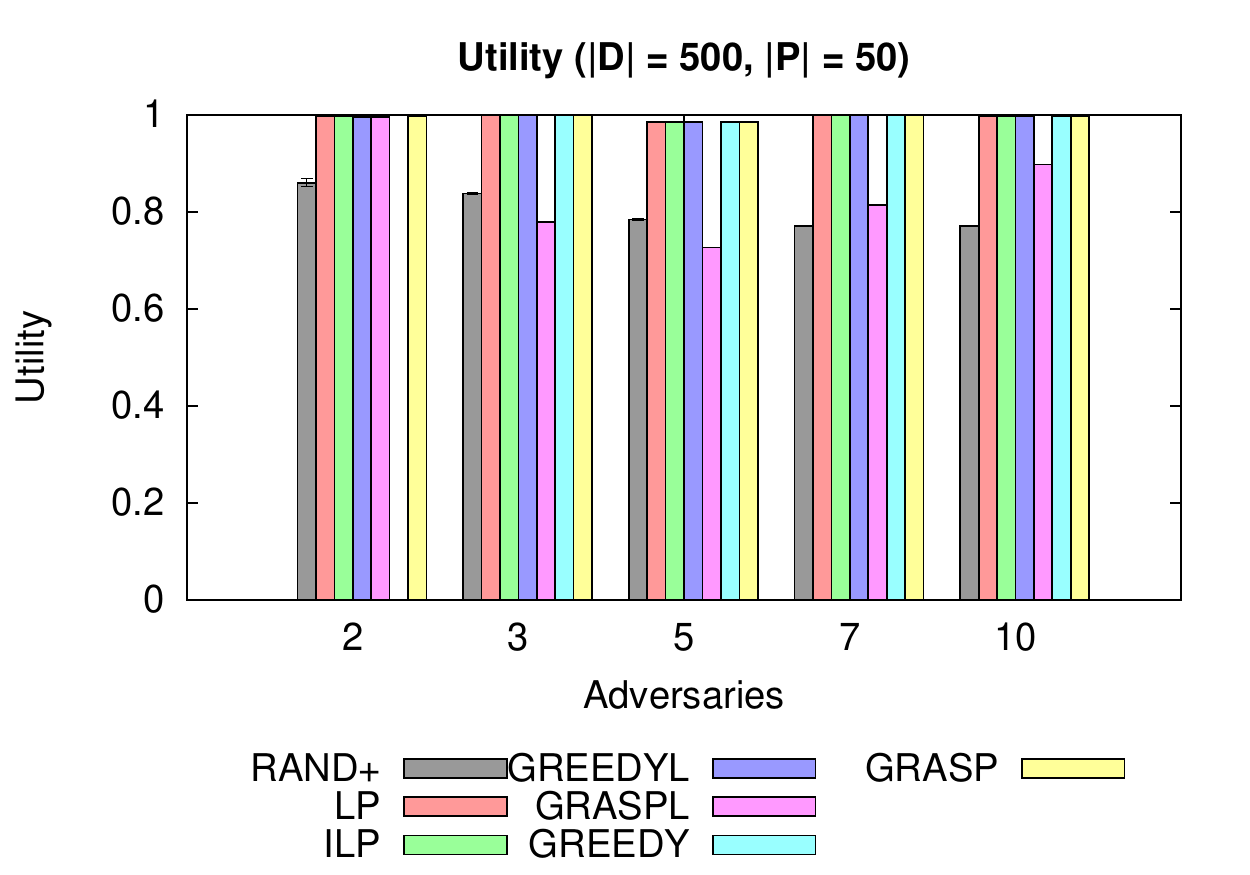} 
\end{center}
\caption{Utility for step disclosure functions. All reported numbers correspond to no information disclosure, while missing values to full disclosure.}
\label{fig:step_fixed_data}
\end{figure}

The corresponding results are presented in Table \ref{tab:step_fixed_prop}. As the number of data entries was increasing for a fixed number of sensitive properties, the number of fully disclosed properties started decreasing. This behavior is expected as the number of data entries per single property increases, and hence, it is easier for our heuristics to find a partitioning where the data items for a single property are partitioned across adversaries inducing zero disclosure.

\begin{table}[h]
\small \centering
\caption{Fully disclosed properties for data-to-adversary assignments for step disclosure functions. As the number of data entries per property increases, local-search heuristics exploit the structure of the problem better and report solutions with fewer fully disclosed properties. }
\begin{tabular}{| c || c || c || c || c |}
\hline
\multicolumn{5}{|c|}{Number of fully disclosed properties ($|P|$ = 50, k = 2)} \\
\hline
$Alg.$ & $|D|$=100 & $|D|$=200 & $|D|$=300 & $|D|$=500\\
\hline
GREEDYL & 11 & 6 & 3 & 0\\
GRASPL & 11& 6 & 4 & 0\\
GREEDY & 11 & 6 & 3 & 1\\
GRASP & 11 & 6 & 5 & 0\\
\hline
\end{tabular}
\label{tab:step_fixed_prop}
\end{table}



The experiments above show that GREEDY and GRASP are viable alternatives for the case of step functions. However, for harder instances of the problem, where the number of sensitive properties is large and the number of adversaries is small, solving the LP relaxation offers a more robust and reliable alternative.

\textbf{Linear Functions.} We continue our discussion and present our experimental results for linear functions. First, we compared the quality of solutions produced when solving \discbudget and \tradeoff optimally for both worst and average disclosure. We generated a synthetic instance of the problem by setting $|D| = 50$ and $|P| = 10$, and we run ILP for $|A| = \{2,3,5,7,10\}$. We set the maximum allowed disclosure to $\tau_I = 0.9$ and $t = 2$.

\begin{figure}[b]
\begin{center}
\subfigure{\includegraphics[trim=50 0 60 0,clip,scale=0.46]{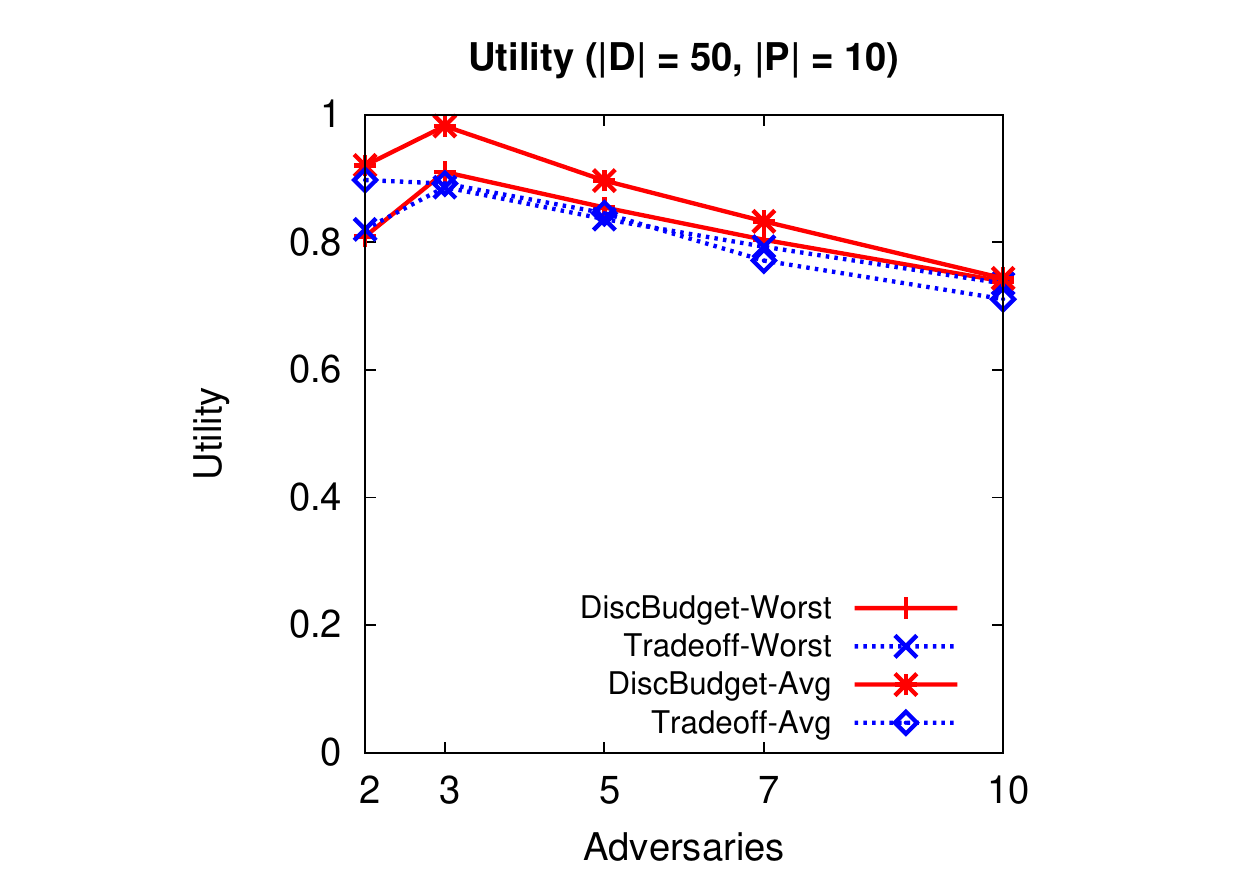} \label{fig:budget_trade_util}}
\subfigure{\includegraphics[trim=50 0 57 0,clip,scale=0.46]{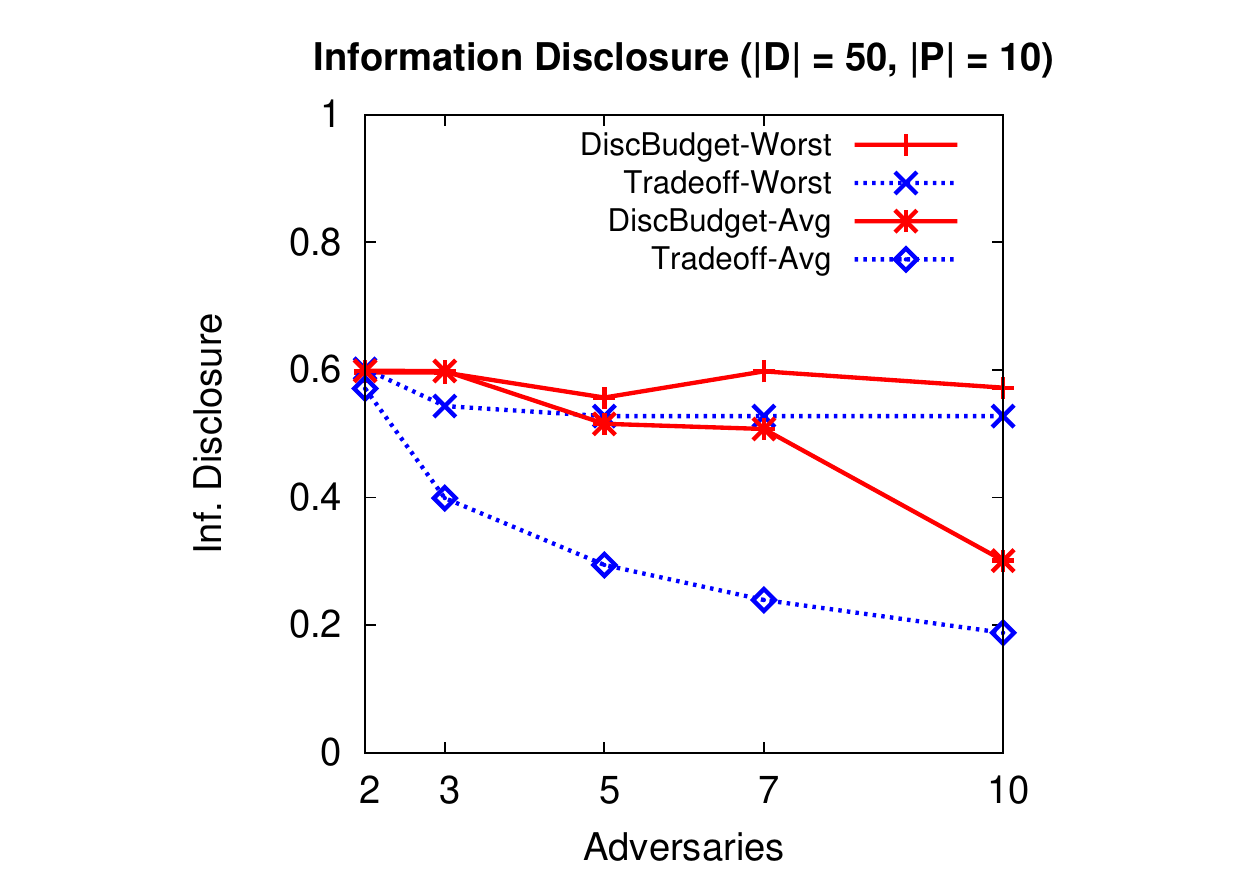} \label{fig:budget_trade_disc}}
\end{center}
\caption{The (a) utility and (b) disclosure when solving \discbudget and \tradeoff optimally. \tradeoff can exploit the presence of multiple adversaries better, to reduce disclosure while maintaining high utility.}
\label{fig:budget_vs_trade}
\end{figure}

\begin{figure*}
\begin{center}
\subfigure{\includegraphics[trim=55 0 62 0,clip,scale=0.5]{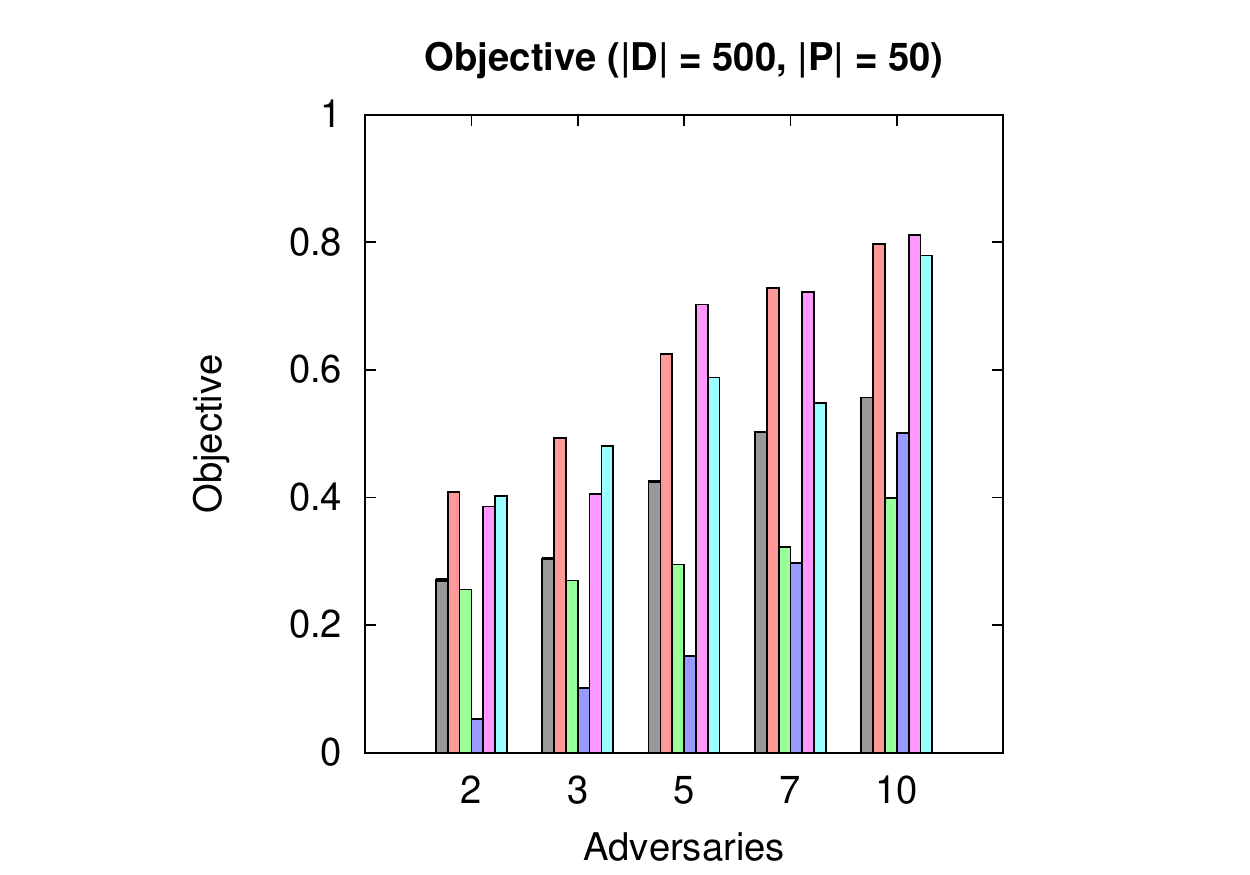} \label{fig:linear_prop_worst_obj}}
\subfigure{\includegraphics[trim=56 0 62 0,clip,scale=0.5]{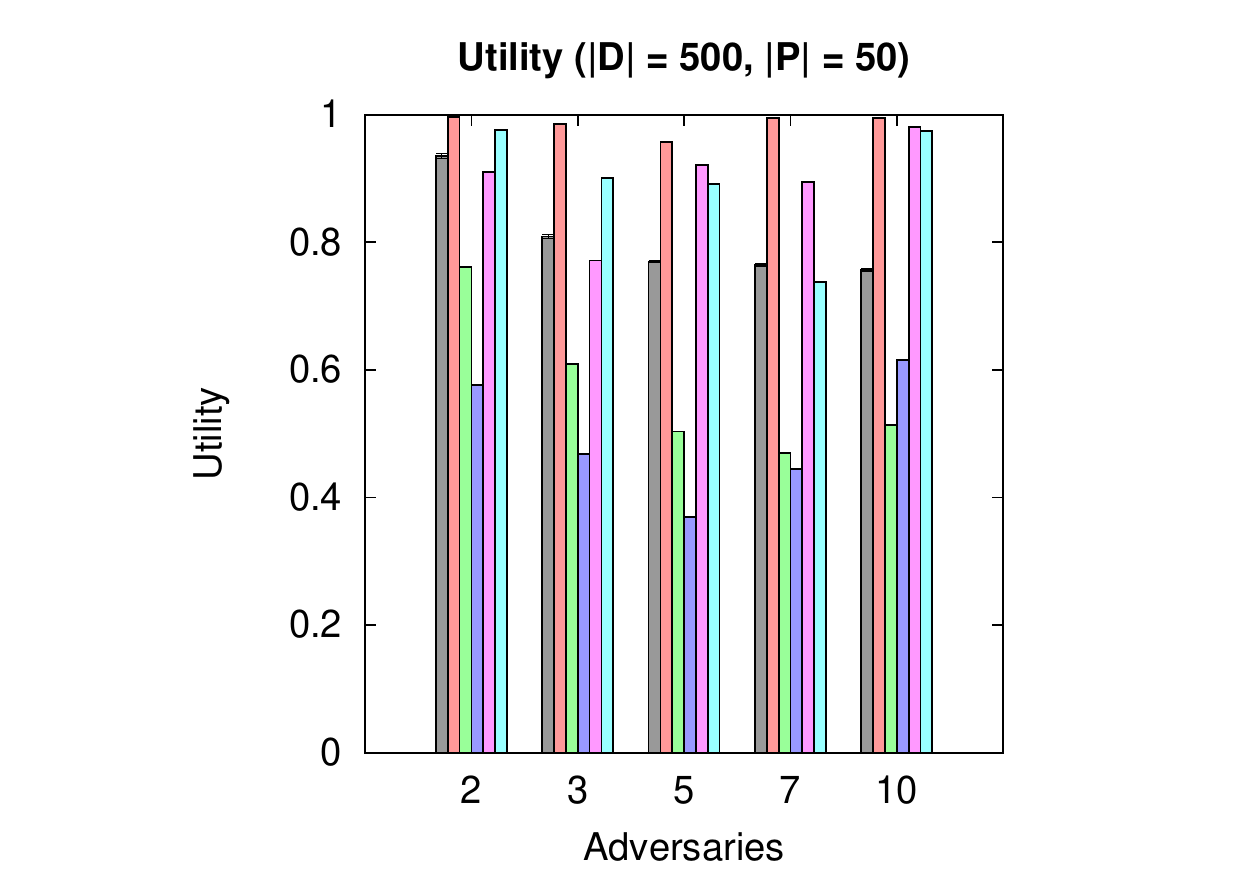} \label{fig:linear_prop_avg_util}}
\subfigure{\includegraphics[trim=10 0 10 0,clip,scale=0.5]{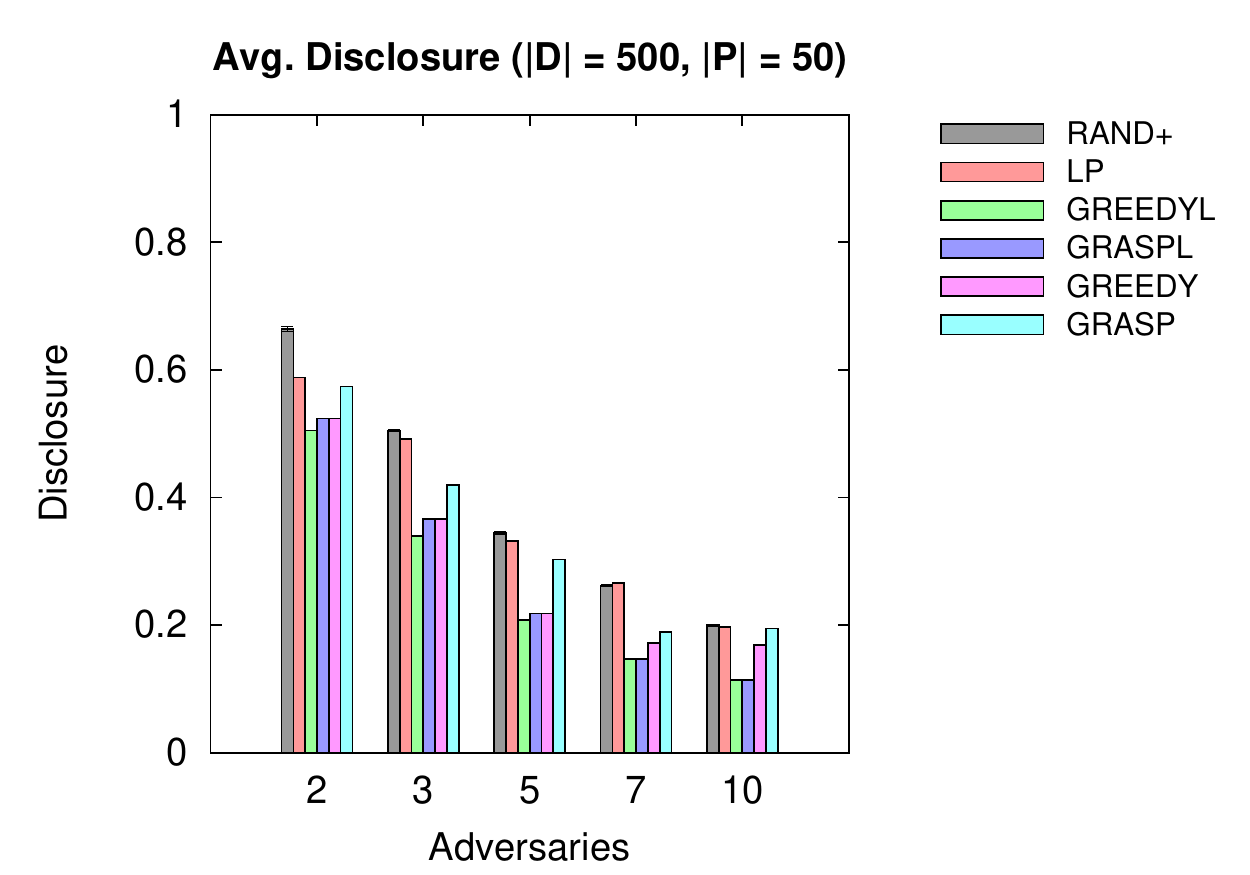} \label{fig:linear_prop_avg_disc}}
\end{center}
\caption{Tradeoff objective, utility and disclosure for linear functions considering average disclosure ($|D| = 500$, $|P| = 50$).  LP, GREEDY and GRASP outperform RAND+ both in terms of total utility and average disclosure. LP maximizes utility, while the local-search heuristics are the most effective in minimizing disclosure.}
\label{fig:linear_prop_worst_avg}
\end{figure*}

The utility and corresponding disclosure are shown in Figure \ref{fig:budget_vs_trade}  when \discbudget and \tradeoff are solved optimally. As shown, the worst case disclosure remains at the same levels across different number of adversaries. Now, consider the case of average disclosure (see Equation \ref{eq:total_disclosure}). We see that for both versions of the optimization problem the disclosure is decreasing as the number of adversaries increases. However, the optimization corresponding to \tradeoff is able to exploit the presence of multiple adversaries better, to reduce disclosure while maintaining utility high. 

%
%

Consequently, we evaluated RAND+, LP, GREEDYL, GRASPL, GREEDY and GRASP on solving \tradeoff . We set upper disclosure $\tau_I = 0$ denoting our requirement to minimize disclosure as much as possible. We do not report any results for the ILP since for $|D| > 50$, it was not able to terminate in reasonable time. First, we fixed the number of properties to $|P| = 50$ and considered instances with $|D| = \{100,200,300,500\}$. We considered average disclosure. We show the performance of the algorithms for $|D| = 500$ in Figure \ref{fig:linear_prop_worst_avg}. As shown, LP, GREEDY and GRASP outperform RAND+ both in terms of the overall utility and the average disclosure. In fact RAND+ performs poorly as it returns solutions with higher disclosure but lower utility than the LP. Furthermore, we see that the performance gap between the local-search algorithms and RAND+ keeps increasing as the number of adversaries increases. This is expected as the proposed heuristics take into account the structure of the underlying dependency graph, and, hence, can exploit the presence of multiple adversaries to achieve a higher overall utility and lower disclosure.

Furthermore, we see that solutions derived using the LP approximation provide the largest utility, while solutions derived using the proposed local-search algorithms minimize the disclosure. As presented in Figure \ref{fig:linear_prop_worst_avg} using the myopic construction returns solutions with low overall utility. This behavior is expected since the algorithm does not maintain a global view of the data-to-adversary assignment. Observe that randomization improves the quality of the solution with respect to total utility, when the global-view construction is used (see GREEDY and GRASP). When the myopic construction is used, randomization gives lower quality solutions.

Measuring the average disclosure is an indicator for the overall performance of the proposed algorithms. However, it does not provide us with detailed feedback about the information disclosure across properties for the different algorithms. To understand what is the exact information disclosure for the solutions returned by the different algorithms, we measured the total number of properties that exceeded a particular disclosure level.  We present the corresponding plots for $|D| = 500$, $|P| = 50$, $k = 2$ and $k = 7$ in Figure \ref{fig:disc_levels}. As shown, the proposed search-algorithms can exploit the presence of multiple adversaries very effectively in order to minimize disclosure. If we compare Figures \ref{fig:disc_level_2} and \ref{fig:disc_level_2} we see that the total number of properties reaches zero for a significantly smaller disclosure threshold in the presence of ten adversaries.
\begin{figure}[h]
\begin{center}
\subfigure{\includegraphics[trim=56 0 62 0,clip,scale=0.45]{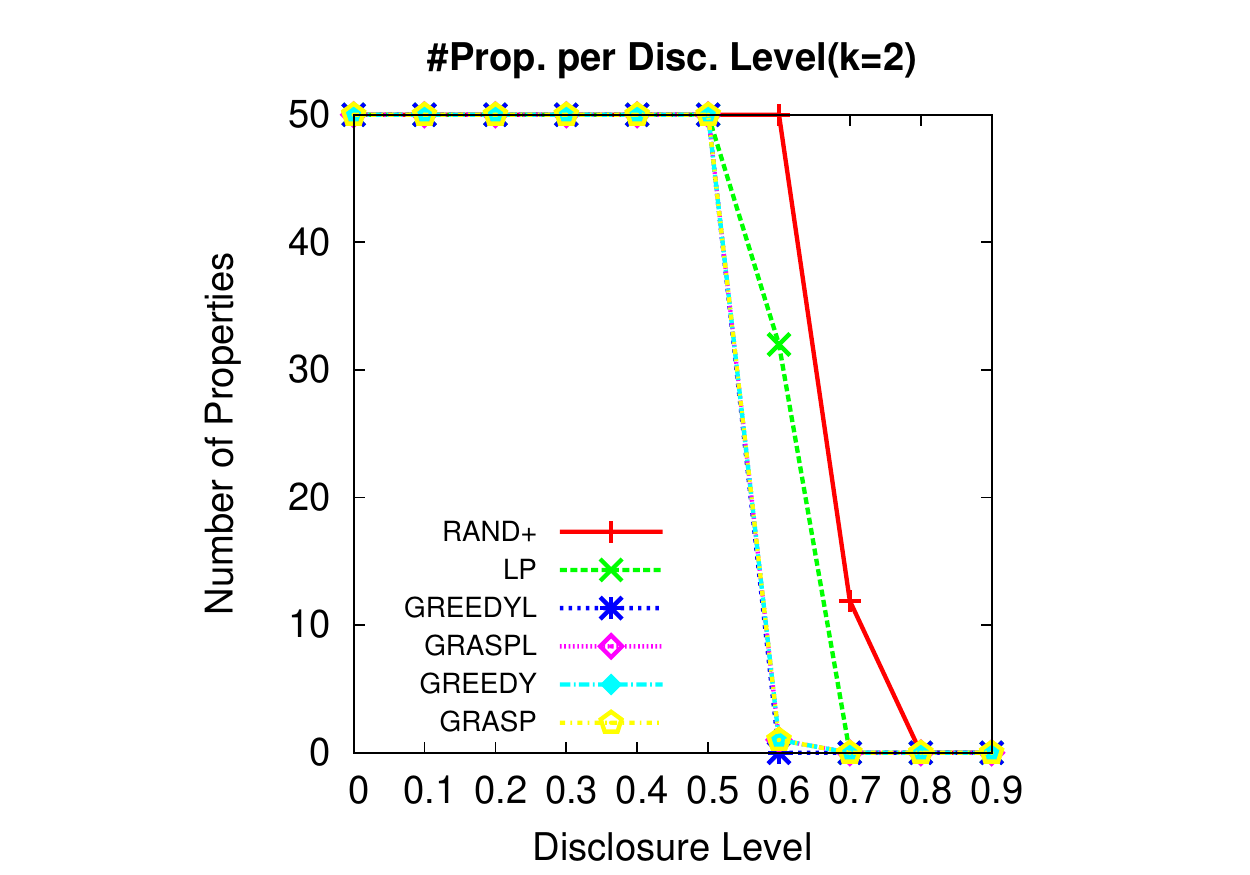} \label{fig:disc_level_2}}
\subfigure{\includegraphics[trim=56 0 62 0,clip,scale=0.45]{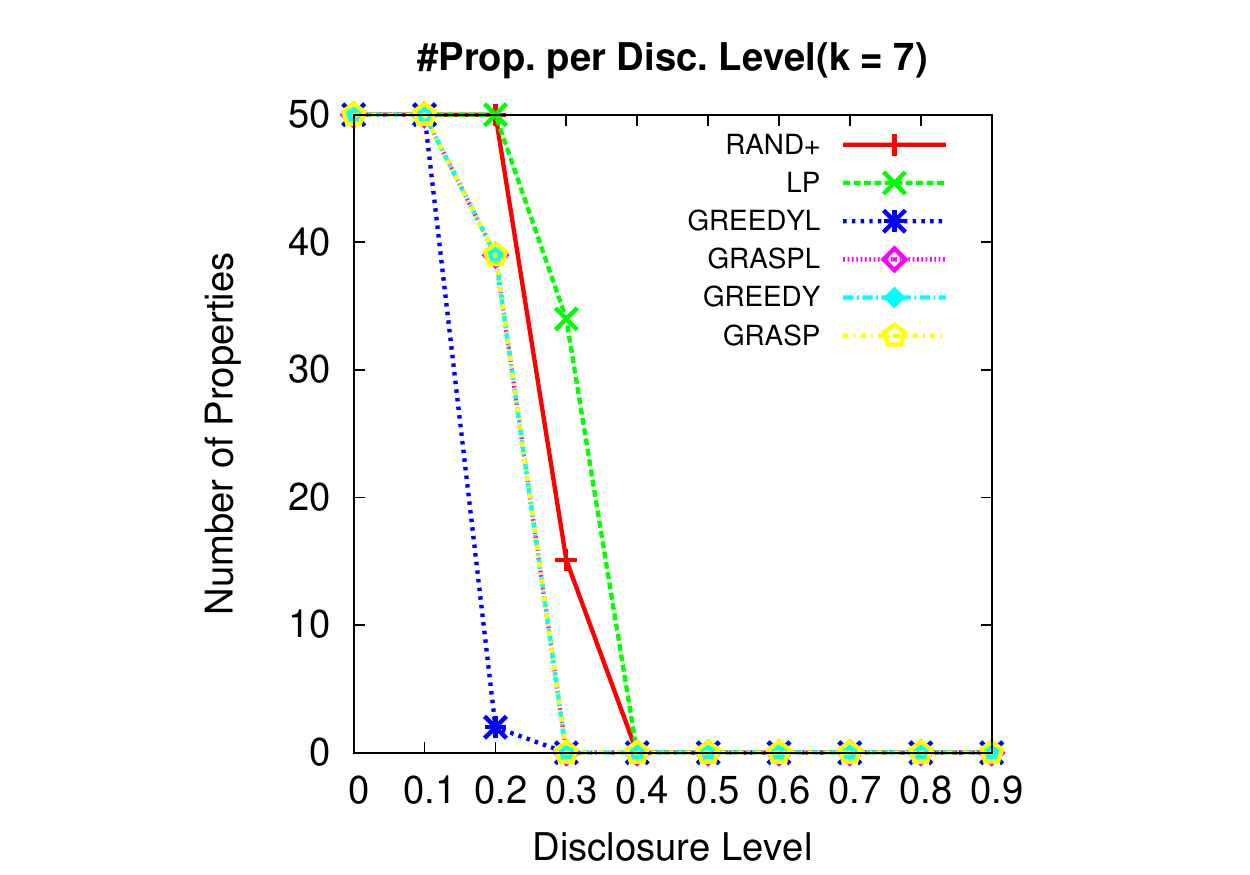} \label{fig:disc_level_7}}
\end{center}
\caption{The number of properties that exceed a particular disclosure level for $|D| = 500$ and $|P| = 50$. GREEDYL and GRASPL can exploit the presence of multiple adversaries more effectively to minimize disclosure. The total number of properties reaches zero for a significantly smaller disclosure threshold in  the presence of ten adversaries.}
\label{fig:disc_levels}
\end{figure}

\section{Related Work}

There has been much work on the problem of publishing (or allowing aggregate queries over) sensitive datasets (see surveys \cite{chen09:fnt,dworksurvey}). Here, information disclosure is characterized by a privacy definition, which is either syntactic constraints on the output dataset (e.g., $k$-anonymity \cite{sweeney02:kAnon} or $\ell$-diversity \cite{ashwin06:ldiversity}), or constraints on the publishing or query answering algorithm (e.g., $\epsilon$-differential privacy \cite{dworksurvey}). Each privacy definition is associated with a privacy level ($k$, $\ell$, $\epsilon$, etc.) that represents a bound on the information disclosure. Typical algorithmic techniques for data publishing or query answering, which include generalization, or coarsening of values, suppression, output perturbation, and sampling, attempt to maximize the utility of the published data given some level of privacy (i.e., a bound on the disclosure). Krause et al. \cite{krause:2008} consider the problem of trading-off utility for disclosure, and consider general submodular utility and supermodular disclosure functions. This paper formulates a submodular optimization problem, and presents efficient algorithm for the same. However, all the above techniques assume that all the data is published to a single adversary. Even when multiple parties may ask different queries, prior work makes a worst-case assumption that they arbitrarily collude. On the other hand, in this paper, we formulate the novel problem of multiple non-colluding adversary, and develop near-optimal algorithms for trading-off utility for information disclosure in this setting.



%
%
%
\section{Conclusions and Future Work}
More and more sensitive information is released on the Web and processed by online services, naturally raising concerns related to privacy in domains where detailed and fine-grained information must be published. In this paper, motivated by applications like online advertising and crowd-sourcing markets, we introduce the problem of privacy-aware $k$-way data partitioning, namely, the problem of splitting a sensitive dataset among $k$ untrusted parties. We present \sparsi a theoretical framework that allows us to formally define the problem as an optimization of the tradeoff between the utility derived by publishing the data and the maximum information disclosure incurred to any single adversary. Moreover, we prove that solving it is NP-hard by reducing it to hypergraph partitioning. We present a performance analysis of different approximation algorithms for a variety of synthetic and real-world datasets, and demonstrate how \sparsi can be applied in the domain of online advertising. Our algorithms are able to partition user-location data to multiple advertisers while ensuring that almost no sensitive information about potential friendship links about these users can be inferred by any advertiser.

Our research so far has raised several interesting research directions. To our knowledge, this is the first work that leverages the presence of multiple adversaries to
minimize the disclosure of private information while maximizing utility. While we provided worst case guarantees for several families of disclosure functions, an interesting future direction is to examine if rigorous guarantees can be provided for other widely-used information disclosure functions like information gain, or if the current ones can be improved. Finally, it is of particular interest to consider how the proposed framework can be extended to consider interactive scenarios where data are published to adversaries more than once, or in streaming data where the partitioning must be done in an online manner.

\bibliographystyle{abbrv}
\bibliography{vldb_secure_partitioning}
\end{document}